%% file: Quantum_no_colors_revised_manuscript.tex
\documentclass[a4paper,aps,prx,twocolumn,superscriptaddress,floatfix,citeautoscript,preprintnumbers,nobalancelastpage,allowtoday,accepted=2026-03-17]{quantumarticle}
\pdfoutput=1
\usepackage[utf8]{inputenc}

\newcommand{\marta}[1]{#1}

\input{packages}

\begin{document}

\title{Regular language quantum states}

\author{Marta Florido-Llin\`as}
\email{marta.florido.llinas@mpq.mpg.de}
\affiliation{
Max-Planck-Institut f{\"{u}}r Quantenoptik, Hans-Kopfermann-Str. 1, 85748 Garching, Germany
}
\affiliation{
Munich Center for Quantum Science and Technology (MCQST), Schellingstr. 4, 80799 M{\"{u}}nchen, Germany
}

\author{\'Alvaro M. Alhambra}
\affiliation{
Instituto de F\'isica Te\'orica UAM/CSIC, C/ Nicol\'as Cabrera 13-15, Cantoblanco, 28049 Madrid, Spain}

\author{David P\'erez-Garc\'ia}
\affiliation{Departamento de An\'alisis Matem\'atico, Universidad Complutense de Madrid, 28040 Madrid, Spain}

\author{J. Ignacio Cirac}
\affiliation{
Max-Planck-Institut f{\"{u}}r Quantenoptik, Hans-Kopfermann-Str. 1, 85748 Garching, Germany
}
\affiliation{
Munich Center for Quantum Science and Technology (MCQST), Schellingstr. 4, 80799 M{\"{u}}nchen, Germany
}%

\date{\today}

\begin{abstract}
We introduce regular language states, a family of quantum many-body states. They are built from a special class of formal languages, called regular, which has been thoroughly studied in the field of computer science. They can be understood as the superposition of all the words in a regular language and encompass physically relevant states such as the GHZ-, W- or Dicke-states. By leveraging the theory of regular languages, we develop a theoretical framework to describe them. First, we express them in terms of matrix product states, providing efficient criteria to recognize them. We then develop a canonical form which allows us to formulate a fundamental theorem for the equivalence of regular language states, including under local unitary operations. We also exploit the theory of tensor networks to find an efficient criterion to determine when regular languages are shift-invariant.
\end{abstract}

\maketitle

\section{Introduction}

The identification of physically relevant families of many-body quantum states admitting a concise theoretical description has played a crucial role in the development of quantum science. Prominent examples of such families are stabilizer- \cite{gottesman_1997_stabilizers}, graph- \cite{hein_2004_graphstates}, Gaussian- \cite{weedbrook_2012_gaussian_qinfo}, or tensor network-states \cite{orus_tensor_2019}, which have turned into fundamental tools across diverse areas such as quantum error correction, measurement-based quantum computation, the classification of phases of matter, and quantum communication \cite{shor_1995_qec, terhal_2015, raussendorf_2001, briegel_2009_measurement-based, Cirac_2021_review}. These families, among others, have allowed to establish unexpected but tight connections across different fields, ranging from quantum information theory, to condensed matter or quantum gravity \cite{zeng_quantum_2019, ryu_2006, evenbly_tensor_2011, Cirac_2021_review}.

In this work, we introduce a family of quantum many-body states that we call \textit{regular language states} (RLS). This family is inspired by the concept of formal languages, which lie at the heart of theoretical computer science, and consist of collections of strings described by simple mathematical rules \cite{Aho_Ullman_1972_book, sipser_introduction_1997, Hopcroft_2001_book_automata}. \marta{Formal languages provide the basic framework underlying modern models of computation, parsing, and program semantics; their hierarchy (regular, context-free, Turing decidable, etc.) organizes computational tasks by the amount of memory or control structure they require. Because of this, they form a unifying framework for describing both classical and quantum computational processes, ranging from compiler design to the theory of quantum finite automata and quantum regular grammars.}

A fundamental subset of formal languages is the class of the so-called regular languages \cite{Yu_1997_regular_languages}, which are widely used for efficient pattern manipulation in applications such as text processing, search engines, and compiler design \cite{Aho_1990_patterns, friedl_2006_regex, aho_compilers_2014}. The structural simplicity of regular languages naturally leads us to define sets of quantum states associated to them as the superposition of the words in the language for each given length. It also allows us to study their properties by leveraging the vast range of tools that have been developed for regular languages in the aforementioned contexts. This definition is also motivated by the fact that prominent examples of states appearing in quantum information and condensed matter physics are RLS, such as the GHZ- \cite{Greenberger_1989}, W- \cite{Dur_2000_Wstate}, or Dicke-states \cite{Dicke_1954}. 

Regular languages can be generated by finite automata \cite{Hopcroft_2001_book_automata, Sakarovitch2009}, which are a special type of restricted Turing machines that follow a simple set of rules. \marta{Finite automata are the simplest and most well-understood memory-bounded computational models: they capture exactly those tasks that can be carried out with constant internal memory, independent of the input size. Their transparency and mathematical rigidity have made them central not only in classical computation, but also in quantum computational models such as quantum finite automata, interactive protocols, and streaming algorithms \cite{bhatia_2019_quantumfiniteautomatasurvey, ambainis_2018_automataquantumcomputing, nayak_2022_augmentedindexquantumstreaming}. For our purposes, automata provide a compact and operational way to encode the structure of a regular language, and therefore of the corresponding RLS. 
As we will see, this connection allows one to express RLS as matrix product states (MPS), the one-dimensional class of tensor network states.}

While the theoretical description of MPS is well established by now \cite{fannes_finitely_1992,Perez-Garcia2007, Cirac_2021_review}, it turns out that it cannot be applied to general RLS. Indeed, two basic tools, the canonical form and the fundamental theorem, require some technical conditions that are not met by such states. The former one removes the ambiguity in the description in terms of MPS, while the latter allows one to recognize if two MPS are LU-equivalent; that is, if they can be related to each other by local unitary operators. Both play a very relevant role in many applications of MPS, ranging from computational methods \cite{Schollwock_2011, zhang_2020_stabilitytensornetworkscanonical} to the classification of symmetry-protected topological phases of matter in one dimension \cite{Chen_2011_symm, Schuch_2011_symm}. Here we combine and extend some of the tools and techniques devised for regular languages and MPS in order to develop a theory that encompasses RLS. Specifically, we define a canonical form that can be found using RL tools, and a fundamental theorem that allows us to analyze the LU-equivalence of RLS. We also illustrate with examples how, in certain aspects, RLS significantly differ from standard MPS. 

\section{Regular language quantum states} \label{sec:RLS}


\marta{Regular languages form one of the most fundamental classes in formal language theory, and their structure underlies the quantum states introduced in this work.}

Given an alphabet $\Sigma$, a \textit{word} is any finite sequence of letters of $\Sigma$, and $\Sigma^*$ denotes the set of all possible words of arbitrary length. We will take $\Sigma=\{0,1,\ldots,d-1\}$, so that $\Sigma^*$ is the set of all sequences composed of those numbers. A language, $L$, is simply a subset of $\Sigma^*$. \marta{There exist different types of languages depending on how they are characterized, but \emph{regular languages} (RL) form the simplest nontrivial class. They admit two fully equivalent characterizations: one relies on an algebraic viewpoint, and the other on an automata-theoretic viewpoint.}

\marta{From the algebraic perspective, regular languages are defined as the smallest class of languages built from the basic languages $\emptyset$, $\varepsilon$ (the \textit{empty string}) and single letters $a \in \Sigma$, which is closed under the so-called \textit{regular operations}: concatenation ($L_1 L_2$), union ($L_1 \cup L_2$) and Kleene star ($L^* := \varepsilon \cup L \cup LL \cup \dots$). The expressions obtained under these rules are known as \emph{regular expressions}. Therefore, they are patterns denoting (possibly infinite) sets of strings, whose manipulation and memory consumption are computationally very efficient, which explains their widespread use in applications such as pattern matching and text processing \cite{friedl_2006_regex}.}

For $\Sigma = \{0,1\}$ ($d=2$) some examples are (\textit{a}) $0^* \cup 1^*$, including all words composed entirely of 0s or entirely of 1s, (\textit{b}) $0^* 1 1^*$, with all words consisting of some 0's followed by one or more 1's, (\textit{c}) $0^* 1 1^* 0 0^*$, with all words with a block of 0's, then a nonempty block of 1's and then again a block of 0's, i.e. all strings $0^i 1^j 0^k$ where $i \geq 0$, $j,k \geq 1$, (\textit{d}) $0^* 1 0^*$, with all words with exactly one 1 and the rest 0s, and (\textit{e}) $0^*10^*10^*$, with all words with exactly two ones in any locations and the rest 0s.

\marta{Equivalently, regular languages can be defined from an automata-theoretic viewpoint.} In this picture, a regular language is precisely one that can be recognized by a \emph{finite automaton}, \marta{a device with a finite set of internal memory states and a set of transition rules specifying how it reacts to each symbol of the input}. Formally, they are structures characterized by a tuple $\mathcal{F} = \langle Q, \Sigma, \delta, I, F\rangle$, where $Q$ is the set of internal states, $I, F \subseteq Q$ are the sets of initial and accepting states, respectively, and $\delta : Q \times (\Sigma \cup \varepsilon) \to \mathcal{P}(Q)$ is the transition function, where $\mathcal{P}(Q)$ denotes the power set of $Q$. 

\marta{A word is \textit{accepted} by the automaton if there exists at least one path through the automaton, following these transitions, that starts in an initial state and ends in an accepting one.} Formally, a language is said to be \textit{accepted} by $\mathcal{F}$ if for any word $w = x_1 x_2 \dots x_n \in L$, there is at least one path along the NFA starting at some $r_0 \in I$ and ending in $r_n \in F$, such that $r_0 \xrightarrow{x_1} r_1 \xrightarrow{x_2} \dots \xrightarrow{x_n} r_n$, where $r_{i+1} \in \delta(r_i, x_{i+1})$ \cite{sipser_introduction_1997}. 

Given $\Sigma = \{0,1\}$, we can consider as a first example the language $L := 1^* (011^*)^*$ of all words where every 0 is followed by at least one 1. It is accepted by the automaton $\mathcal{F} = \langle Q, \Sigma, \delta, I, F \rangle$, depicted below.
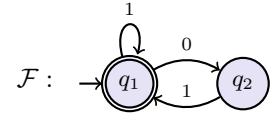
\begin{equation*}
    \mathcal{F}: \begin{tikzpicture}[scale=.2, baseline={([yshift=-3ex]current bounding box.center)}, thick]
        \tikzstyle{small state} = [state, minimum size=0pt, fill=purple,node distance=1.5cm,initial distance=1.5cm,initial text=$ $]
        \tikzset{->}

        \node[small state, initial, accepting] (q1) {\small $q_1$};
        \node[small state, right of=q1] (q2) {\small $q_2$};

        \draw (q1) edge[loop above] node{\scriptsize $1$} (q1)
        (q1) edge[bend left, above] node{\scriptsize $0$} (q2)
        (q2) edge[bend left, above] node{\scriptsize $1$} (q1);
    \end{tikzpicture}\ 
\end{equation*}
The short arrow marks the initial state, $I = \{1\}$. An arrow going from state $q_i$ to state $q_j$, with the symbol $x$ on top, indicates that $j \in \delta(i, x)$. For instance, we have $\delta(1, 0) = \{2\}$, meaning that if at some step of a computation we are in state $q_1$, and emit output symbol $0$, then we transition to $q_2$. Similarly, $\delta(1,1) = \{1\}$ and $\delta(2,1) = \{1\}$. The double circle indicates the accepting state, $F = \{1\}$. 

\marta{As further examples, we depict below automata accepting each of the regular expressions \textit{(a)-(e)} introduced above:}
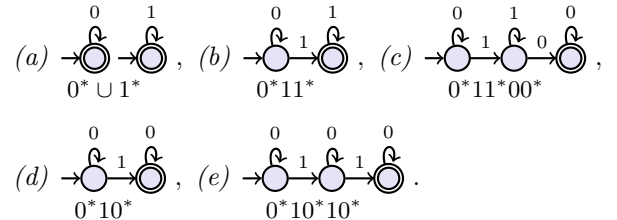
\begin{align*}
    &\textit{(a) } \hspace{-0.2cm}
    \begin{tikzpicture}[scale=.5, baseline={([yshift=-1.2ex]current bounding box.center)}, thick]
        \tikzstyle{small state} = [state, minimum size=7pt, fill=purple,node distance=0.75cm,initial distance=0.5 cm,initial text=]
        \tikzset{->}
        \node[small state, initial, accepting] (q1) {};
        \node[small state, initial, accepting, right of=q1] (q2) {};
        \draw (q1) edge[loop above] node{\scriptsize $0$} (q1)
        (q2) edge[loop above] node{\scriptsize $1$} (q2);
        \node[below=-0.05 cm of current bounding box.south, align=center]
        {\small $0^* \cup 1^*$};
    \end{tikzpicture}
    \ , \
    \textit{(b) } \hspace{-0.2cm} 
    \begin{tikzpicture}[scale=.5, baseline={([yshift=-1.2ex]current bounding box.center)}, thick]
        \tikzstyle{small state} = [state, minimum size=7pt, fill=purple,node distance=0.75cm,initial distance=0.5 cm,initial text=]
        \tikzset{->}
        \node[small state, initial] (q1) {};
        \node[small state, accepting, right of=q1] (q2) {};
        \draw (q1) edge[loop above] node{\scriptsize $0$} (q1)
        (q2) edge[loop above] node{\scriptsize $1$} (q2)
        (q1) edge[above] node{\scriptsize $1$} (q2);
        \node[below=-0.05 cm of current bounding box.south, align=center]
        {\small $0^* 1 1^*$};
    \end{tikzpicture}
    \ , \ \textit{(c) } \hspace{-0.2cm} 
    \begin{tikzpicture}[scale=.5, baseline={([yshift=-1.2ex]current bounding box.center)}, thick]
        \tikzstyle{small state} = [state, minimum size=7pt, fill=purple,node distance=0.75cm,initial distance=0.5 cm,initial text=]
        \tikzset{->}
        \node[small state, initial] (q1) {};
        \node[small state, right of=q1] (q2) {};
        \node[small state, accepting, right of=q2] (q3) {};
        \draw (q1) edge[loop above] node{\scriptsize $0$} (q1)
        (q2) edge[loop above] node{\scriptsize $1$} (q2)
        (q3) edge[loop above] node{\scriptsize $0$} (q3)
        (q1) edge[above] node{\scriptsize $1$} (q2)
        (q2) edge[above] node{\scriptsize $0$} (q3);
        \node[below=-0.05 cm of current bounding box.south, align=center]
        {\small $0^* 1 1^* 0 0^*$};
    \end{tikzpicture} \ , \\ 
    &\textit{(d) } \hspace{-0.2cm} 
    \begin{tikzpicture}[scale=.5, baseline={([yshift=-1.2ex]current bounding box.center)}, thick]
        \tikzstyle{small state} = [state, minimum size=7pt, fill=purple,node distance=0.75cm,initial distance=0.5 cm,initial text=]
        \tikzset{->}
        \node[small state, initial] (q1) {};
        \node[small state, accepting, right of=q1] (q2) {};
        \draw (q1) edge[loop above] node{\scriptsize $0$} (q1)
        (q2) edge[loop above] node{\scriptsize $0$} (q2)
        (q1) edge[above] node{\scriptsize $1$} (q2);
        \node[below=-0.05 cm of current bounding box.south, align=center]
        {\small $0^* 1 0^*$};
    \end{tikzpicture} \ , \
    \textit{(e) } \hspace{-0.2cm} 
    \begin{tikzpicture}[scale=.5, baseline={([yshift=-1.2ex]current bounding box.center)}, thick]
        \tikzstyle{small state} = [state, minimum size=7pt, fill=purple,node distance=0.75cm,initial distance=0.5 cm,initial text=]
        \tikzset{->}
        \node[small state, initial] (q1) {};
        \node[small state, right of=q1] (q2) {};
        \node[small state, accepting, right of=q2] (q3) {};
        \draw (q1) edge[loop above] node{\scriptsize $0$} (q1)
        (q2) edge[loop above] node{\scriptsize $0$} (q2)
        (q3) edge[loop above] node{\scriptsize $0$} (q3)
        (q1) edge[above] node{\scriptsize $1$} (q2)
        (q2) edge[above] node{\scriptsize $1$} (q3);
        \node[below=-0.05 cm of current bounding box.south, align=center]
        {\small $0^* 1 0^* 1 0^*$};
    \end{tikzpicture} \ .
\end{align*}

The equivalence between the two viewpoints is known as Kleene's theorem \cite{Kleene_1956} and is one of the most fundamental results in automata theory: it asserts that for every regular expression there exists an NFA that accepts exactly the same words, and vice versa.

We will be interested in a subset of NFA's called \textit{unambiguous finite automata} (UFA) \cite{Sakarovitch2009}, in which any word has at most one accepting path. In fact, for any RL there always exists an UFA that accepts all and only its words \cite{Ravikumar_1989_succint-ambiguity, Leung_2005_expsep-NFA-UFA}.

\marta{The automata viewpoint also provides intuition about which languages are regular. Since finite automata have only a constant amount of memory, they cannot count or store unbounded information about the input. For example, the language $L = \{0^n 1^n \mid n \in \mathbb{N}\}$ is not regular because accepting it requires remembering exactly how many zeros have appeared before switching to ones, for arbitrarily long inputs.}

To every regular language we can associate a family of quantum states consisting of the superposition with weights equal to one of all the words of length $N$ in the language. 
\begin{definition}
    Given a RL $L$, the family of regular language states (RLS) associated to $L$ is $L_q := \{\ket{L_N}\}_{N \in \mathbb{N}}$, where
    \begin{equation} \label{eq:def_L}
        \ket{L_N} = \sum_{w \in L \cap \Sigma^N} \ket{w}.
    \end{equation}
\end{definition}

\marta{For the examples \textit{(a)-(e)} introduced above in terms of their regular expressions and automata, the corresponding RLS are: 
(\textit{a}) GHZ states on qubits \cite{Greenberger_1989}, 
(\textit{b}) the simplified ansatz for a domain-wall excitation \cite{Haegeman_2012}, 
(\textit{c}) a state with two domain-wall excitations, 
(\textit{d}) W-states \cite{Dur_2000_Wstate}, and
(\textit{e}) Dicke states with two excitations of the same type \cite{Dicke_1954}. 

Altogether, these RLS examples show that regular languages naturally encode many familiar quantum states from quantum information and quantum many-body physics. Moreover, they also show how RLS can systematically and compactly represent excited states of simple models, such as the transverse-field Ising model: \textit{(a)} corresponds to the ground state, \textit{(b)} and \textit{(c)} represent the first and second excited states in the ferromagnetic phase (i.e. single and double domain-wall excitations), while \textit{(d)} and \textit{(e)} capture the structure of single and double quasiparticle excitations in the paramagnetic phase \cite{Sachdev_2011_book-QPT}.}

Note that, in order to keep a simple notation, we have not normalized the states. Furthermore, we have chosen all the coefficients in front of the states to be identical. One can extend RLS to include complex coefficients, although here we will mainly concentrate on the first case.

\section{MPS representation of RLS}

All the examples of the previous section can be written as MPS with a constant bond dimension, independent of $N$. In order to show that this holds for all RLS, we use a subset of NFA's called unambiguous finite automata (UFA) \cite{Sakarovitch2009}, for which any word has at most one accepting path. In fact, for any RL there always exists an UFA that accepts all and only its words \cite{Ravikumar_1989_succint-ambiguity, Leung_2005_expsep-NFA-UFA}. With this, we can now establish a connection between RLS and MPS as follows.
\begin{restatable}{theorem}{lemmaNFAMPS}
    \label{lemma:NFAMPS}
    Any family of RLS $L_q = \{\ket{L_N}\}$ admits an MPS description with binary entries and constant bond dimension. In particular, given any UFA $\mathcal{F} = \langle Q, \Sigma, \delta, I, F\rangle$ that accepts $L$, then
    \begin{equation} \label{eq:familyRLMPS1}
    \ket{L_N} :=
        \begin{tikzpicture}[scale=.45, baseline={([yshift=-1ex]current bounding box.center)}, thick]
            \FullMPS{0,0}{$A$}{purple}
            \draw[fill=amaranth] (-1.4,0) circle (0.4);
            \node at (-1.4,0) {\scriptsize $v_l$};
            \draw[fill=amaranth] (6.4,0) circle (0.4);
            \node at (6.4,0) {\scriptsize $v_r$};
            \node at (3.2,0.5) {\scriptsize $N$ times};
        \end{tikzpicture},
    \end{equation}
    where the bond dimension is $D=|Q|$, and
    \begin{equation} \label{eq:associate_MPS_to_NFA1}
    \begin{cases}
        \begin{tikzpicture}[scale=.45, baseline={([yshift=-0.5ex]current bounding box.center)}, thick]
            \draw (0,0) -- (0.8,0);
            \draw[fill=amaranth] (0,0) circle (0.4);
            \node at (0,0) {\scriptsize $v_l$};
        \end{tikzpicture} := \sum_{i \in I} \bra{i}, \\
        \begin{tikzpicture}[scale=.45, baseline={([yshift=-0.5ex]current bounding box.center)}, thick]
            \draw (0,0) -- (-0.8,0);
            \draw[fill=amaranth] (0,0) circle (0.4);
            \node at (0,0) {\scriptsize $v_r$};
        \end{tikzpicture}
        := \sum_{f \in F} \ket{f},
    \end{cases}
    \begin{tikzpicture}[scale=.4, baseline={([yshift=-1ex]current bounding box.center)}, thick]
        \MPSTensor{0,0}{$A$}{purple}
        \node at (-1.2,0) {\scriptsize $i$};
        \node at (1.2,0) {\scriptsize $j$};
        \node at (0,1.25) {\scriptsize $x$};
    \end{tikzpicture}
    = \begin{cases}
        1 \text{ if } j \in \delta(i,x) , \\
        0 \text{ otherwise}.
    \end{cases}
\end{equation}
\end{restatable}

The proof is provided in Appendix \ref{app:A}. \marta{The resulting states are such that the coefficients in the computational basis belong to $\{0,1\}$. Additionally, }note that the resulting MPS is composed of a single tensor $A$, although it is not necessarily translationally invariant because it has open boundary conditions (OBC). \marta{Therefore, the entanglement of RLS obeys an area law, since they always admit an MPS representation; in particular, the size of any UFA accepting the RL upper bounds the entanglement of the associated RLS.}

\marta{To illustrate the connection between RLS and MPS established in Theorem \ref{lemma:NFAMPS}, we provide below the explicit MPS representation of two of the examples of RLS introduced in the previous section. Given $\Sigma = \{0,1\}$, we consider the languages $L_1 := 0^* 1 0^*$ of all words containing a single 1, and $L_2 := 1^* (011^*)^*$, with all words where every 0 is followed by at least one 1. They are accepted by
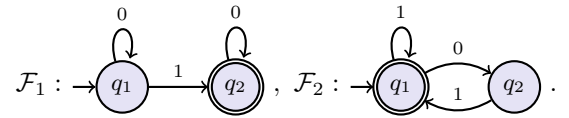
\begin{equation*}
    \mathcal{F}_1:\hspace{-2mm}\begin{tikzpicture}[scale=.2, baseline={([yshift=-3ex]current bounding box.center)}, thick]
        \tikzstyle{small state} = [state, minimum size=0pt, fill=purple,node distance=1.5cm,initial distance=1.5cm,initial text=$ $]
        \tikzset{->}

        \node[small state, initial] (q1) {\small $q_1$};
        \node[small state, accepting, right of=q1] (q2) {\small $q_2$};

        \draw (q1) edge[loop above] node{\scriptsize $0$} (q1)
        (q2) edge[loop above] node{\scriptsize $0$} (q2)
        (q1) edge[above] node{\scriptsize $1$} (q2);
    \end{tikzpicture} \ ,
    \ \mathcal{F}_2:\hspace{-2mm}\begin{tikzpicture}[scale=.2, baseline={([yshift=-3ex]current bounding box.center)}, thick]
        \tikzstyle{small state} = [state, minimum size=0pt, fill=purple,node distance=1.5cm,initial distance=1.5cm,initial text=$ $]
        \tikzset{->}

        \node[small state, initial, accepting] (q1) {\small $q_1$};
        \node[small state, right of=q1] (q2) {\small $q_2$};

        \draw (q1) edge[loop above] node{\scriptsize $1$} (q1)
        (q1) edge[bend left, above] node{\scriptsize $0$} (q2)
        (q2) edge[bend left, above] node{\scriptsize $1$} (q1);
    \end{tikzpicture}\ .
\end{equation*}
Using Eq. \eqref{eq:associate_MPS_to_NFA1}, the MPS associated to the first automaton has the following tensors,
\begin{equation*}
    A^0 = {\footnotesize \begin{pmatrix}
        1 & 0 \\ 0 & 1
    \end{pmatrix}}, \
    A^1 = {\footnotesize \begin{pmatrix}
        0 & 1 \\ 0 & 0
    \end{pmatrix}}, \
    \bra{v_l} = {\footnotesize \begin{pmatrix}
        1 & 0
    \end{pmatrix}}, \
    \ket{v_r} = {\footnotesize \begin{pmatrix}
        0 \\ 1
    \end{pmatrix}},
\end{equation*}
while the automaton for $L_2$ corresponds to
\begin{equation*}
    A^0 = {\footnotesize \begin{pmatrix}
        0 & 1 \\ 0 & 0
    \end{pmatrix}}, \
    A^1 = {\footnotesize \begin{pmatrix}
        1 & 0 \\ 1 & 0
    \end{pmatrix}}, \
    \bra{v_l} = {\footnotesize \begin{pmatrix}
        1 & 0
    \end{pmatrix}}, \
    \ket{v_r} = {\footnotesize \begin{pmatrix}
        1 \\ 0
    \end{pmatrix}}.
\end{equation*}}

Given any tensor $A$ and vectors $v_l, v_r$ with binary entries, is the corresponding MPS a RLS? To answer this, we note that the coefficient in front of each $\ket{w}$ equals the number of accepting paths for $w$, which is not necessarily one unless the automaton is unambiguous. Therefore, the MPS is a RLS if and only if the automaton associated to it according to Eq. \eqref{eq:associate_MPS_to_NFA1} is an UFA.  

Fortunately, RL theory provides a criterion to decide whether an NFA is actually an UFA in time $O(m^2)$, where $m$ is the number of transitions in the underlying NFA \cite{Sakarovitch2009}, which scales at worst as $O(d^2 D^4)$ since $m \leq dD^2$. By expressing this criterion in terms of the algebra of the MPS, consisting of all finite products of the MPS matrices and any linear combinations thereof \cite{farenick_algebras_book_2001}, we present an alternative characterization of unambiguity in the lemma below, which can be checked in time $O(dD^3)$. We denote the algebra of the MPS as $\mathcal{A} := \text{Alg}(\{A^x\}) \subseteq \mathcal{M}_D(\mathbb{C})$, and define the sets $\mathcal{V}_L := \{( a^T \ket{v_l} )^{\otimes 2} \mid a \in \mathcal{A} \}$ and $\mathcal{V}_R := \{(a \ket{v_r})^{\otimes 2} \mid a \in \mathcal{A}\}$, where $a^T$ is the transpose of $a$.

\begin{restatable}{lemma}{UFAMPS} \label{lemma:UFAMPS}  
    Given an MPS defined by the binary tensors $\{A, v_l, v_r\}$ according to Eq. \eqref{eq:familyRLMPS1}, it is a RLS if and only if, for all $m, n \in \{1, \dots, D\}$ with $m \neq n$, 
    \begin{equation*}
        \ket{m} \otimes \ket{n} \in (\mathcal{V}_L)^\perp \cup (\mathcal{V}_R)^\perp.
    \end{equation*}
    This condition can be checked in time $O(dD^3)$.
\end{restatable}

The proof is provided in Appendix \ref{app:B}.

A particularly relevant set of MPS are those that are translationally invariant (TI). Thus, we now address the problem of determining whether a RLS is TI. The first thing to notice is that these states correspond to the so-called shift-invariant languages \cite{Jiraskova_Okhotin_2008}. A naive algorithm based on automata theory for deciding whether a RL is shift-invariant can scale as $O(d D^2 4^{D^2})$ with the bond dimension $D$ of any MPS representation. However, using the theory of MPS through the lemma below, we can tackle this issue more efficiently.

\begin{restatable}{lemma}{lemmaTI}
\label{lemma:TI}
    A family of MPS-X states $\{\ket{\psi_N(X, A)}\}_N$, 
    \begin{equation} \label{eq:general_MPS-X}
        \ket{\psi_N(X, A)} := 
        \begin{tikzpicture}[scale=.45, baseline={([yshift=-1ex]current bounding box.center)}, thick]
            \FullMPSX{(0,0)}{$A$}{$X$}{purple}{yellow}
            \node at (3,0.5) {\tiny $N$ \text{\normalfont times}};
        \end{tikzpicture},
    \end{equation}
    where $X, A^i \in \mathcal{M}_D(\mathbb{C})$ is TI for all $N$, if and only if,
    \begin{equation*}
        \Tr[X[a,b]] = 0, \ \forall a, b \in \text{\normalfont Alg}(\{A^i\}).
    \end{equation*}
\end{restatable}
As a corollary, we demonstrate that shift-invariant RL's can be identified in time $O(dD^3)$. The proofs can be found in Appendix \ref{app:C}. This shows that tensor networks tools can also add a new perspective on the study of certain problems in computer science.
\begin{restatable}{corollary}{corollaryTI}
\label{corollary:TI}
    A RL is shift-invariant if and only if, for any MPS representation of the associated RLS with tensors $\{A, v_l, v_r\}$ of bond dimension $D$, it holds that $\mel{v_l}{[a,b]}{v_r} = 0$ for all $a, b \in \text{\normalfont Alg}(\{A^i\})$. This condition can be checked in time $O(dD^3)$. 
\end{restatable}

\section{Canonical form of RLS}

Given two RLS and their MPS representations with bond dimensions $D_1$ and $D_2$, how can we determine if they are equivalent? RL theory provides an algorithm to solve this problem in time $O(d(D_1 + D_2)^3)$ for the case of UFA's \cite{Schutzenberger_1961, Stearns_1985-UFAs-Equiv_Contain, Cortes_2006_equiv_cubic}, \marta{and even in $O(\max_i dD_i \log D_i)$ when the automata are deterministic (i.e. when each state has at most one outgoing transition per output symbol) \cite{Hopcroft_2001_book_automata}}. \marta{Note that standard MPS methods are ill-suited here: those for fixed string length $N$ with site-dependent tensors scale as $(O(dN(D_1 + D_2)^3)$, while methods for all $N$ with uniform tensors in the bulk assume periodic boundary conditions and therefore do not apply to RLS, whose MPS representations have open boundaries.}

But is it possible to find a canonical form that is efficiently computable and uniquely associated with each RLS? Ideally, this canonical form would enable us not only to determine the equivalence between RLS, but also to address physically relevant questions such as the LU-equivalence of RLS, which is key to understand the presence of symmetries \cite{Chen_2011_symm, Schuch_2011_symm}, or to study the entanglement properties of RLS in terms of their interconvertibility \cite{Kraus_2010_LU, Sauerwein_2019_MPSslocc}.

As a first attempt, we might consider using the already existing canonical form of MPS, but this is not generally possible, because many RLS do not fit into the traditional MPS framework. An example of this fact is the RLS consisting of W-states, for which imposing a TI MPS representation with periodic boundary conditions results in a bond dimension scaling with the system size as $\Omega(N^{1/(3+\delta)})$ for any $\delta > 0$ \cite{Perez-Garcia2007, Michalek2018}. \marta{On the other hand, note that the so-called left or right canonical form in the fixed $N$ setting for non-uniform MPS (which refers to the tensors being left or right isometries \cite{Perez-Garcia2007}) are not applicable to our setting for all $N$.}

In what follows, we propose an alternative representation that enables us to formulate a fundamental theorem for LU-equivalent RLS. For this purpose, we introduce first a canonical decomposition for regular languages, and use it as the basis to define a canonical form for RLS.

Given any RL, $L$, on alphabet $\Sigma$, we will relabel it as $\Sigma = \Sigma_\infty \cup \Sigma_f = \{0, 1, \dots\} \cup \{|\Sigma_\infty|, |\Sigma_\infty|+1, \dots, d-1\}$, where $\Sigma_f$ contains the symbols of $\Sigma$ whose number of appearances in any word of $L$ is upper bounded by a constant. We will also use operation $S^{(m)}: \mathcal{L}_m \times (\Sigma_f)^m \to \Sigma^*$, where $\mathcal{L}_m$ is the set of all RL's in $(\Sigma_\infty \cup \{f\})^*$ with exactly $m$ appearances of the symbol $f$ in any word, such that $S^{(m)}(L_1, L_2)$ substitutes every occurrence of $f$ in $L_1$ with the strings indicated by $L_2$.

\marta{The \emph{tree width} of an automaton measures the maximum number of distinct partial computations that any word can generate. Deterministic automata have tree width $1$, while nondeterministic transitions on cycles can lead to arbitrarily many partial computations and hence infinite tree width. More details are included in Appendix \ref{app:D}.}

\begin{restatable}[Canonical decomposition of a RL]{theorem}{propcanRL} \label{prop:canonical_decomposition_RL}
    Any RL, $L$, can be decomposed in a unique way in terms of the pairs of RL's $\{(L_j^m , X_j^m) \}_{j,m}$, where $m \leq M$, as
    \begin{equation*}
        L = \cup_m \cup_j S^{(m)} (L_j^m , X_j^m) \ , 
    \end{equation*}
    where $L_j^m \in \mathcal{L}_m$ and $X_j^m$ consists of a union of strings of $(\Sigma_f)^m$, with the property that $L_j^m \cap L_k^m = \emptyset$ and $X_j^m \neq X_k^m$ unless $j = k$. 
    
    \marta{This decomposition is efficiently computable: the runtime scales polynomially with $D$ if the input RL $L$ is given as a DFA or an NFA of finite tree width $k$ ($O(D^{2d^M})$ and $O(D^{2kd^M})$, respectively), or exponentially for an NFA of infinite tree width ($O(4^{(M+1)d^M D})$).} 
\end{restatable}

The proof is included in Appendix \ref{app:D}\marta{, where the explicit time scaling of the algorithm with $D, d, M$ is also included.} For instance, $L_1 = 0^*10^*$ can be canonically decomposed as $\{(0^* f 0^*, 1)\}$, so $L_1 = S^{(1)}(0^* f 0^*, 1)$, and $L_2 = 0^* 1 0^* 2 0^* \cup 0^* 2 0^* 1 0^*$ can be canonically decomposed as $\{(0^* f 0^* f 0^*, 12 \cup 21)\}$, so $L_2 = S^{(2)}(0^* f 0^* f 0^*, 12 \cup 21)$.

We now embed $L_j^m$ and $X_j^m$ into quantum states, where $| L_{j,N}^m \rangle$ is the superposition of all strings of length $N$ in $L_j^m$, and similarly $| X_j^m \rangle$ is the sum of all strings in $X_j^m$, which have length $m$ by construction. We also let $\hat{S}^{(m),N}$ be the linear operator such that $\hat{S}^{(m), N}| x \rangle | y \rangle$ is the the sum of the words of length $N$ in $S^{(m)}(x,y)$. The dependence on $N$ will be omitted for ease of notation. 

\begin{corollary} \label{cor:canonical_decomposition_RLS}
    Given a RLS $L_q$, it can be uniquely written in terms of RLS $\{(| L_j^m \rangle, | X_j^m \rangle)\}_{j,m}$ as
    \begin{align*}
        | L^N \rangle &= \sum_m \sum_j \hat{S}^{(m)} | L_j^{m} \rangle | X^m_j \rangle ,
    \end{align*}
    where $L_j^m$ and $X_j^m$ have the same properties as in Theorem \ref{prop:canonical_decomposition_RL}, meaning that $\langle L_j^m | L_k^m \rangle = 0$ and $|X_j^m \rangle \neq |X_k^m\rangle$ unless $j = k$.
\end{corollary}

As the last step towards a canonical form of RLS, we need to look for a canonical choice of the MPS tensors for the RLS $\{(| L^m_j \rangle, | X^m_j \rangle)\}_{j,m}$ in Corollary \ref{cor:canonical_decomposition_RLS}. A RL can be accepted by many UFA's, each being a valid MPS representation according to Theorem \ref{lemma:NFAMPS}. However, we cannot use the traditional canonical form of MPS because it is not generally applicable to RLS, as previously mentioned. 

Automata theory provides a solution for this, by showing that any RL admits a canonical deterministic finite automaton (DFA). This canonical DFA is minimal and unique, up to relabelling of the internal states, among all DFA's accepting the language. Given an UFA of bond dimension $D$, the minimal equivalent DFA can be found in time $O(d D 2^D)$, \marta{or in $O(dD \log D)$ if the automaton is already deterministic} \cite{Hopcroft_1971, Hopcroft_2001_book_automata, Leung_2005_expsep-NFA-UFA}. Note that there might be UFA's or ansätze as the one of Eq. \eqref{eq:general_MPS-X} with rank$(X) \geq 1$ yielding smaller bond dimensions than the minimal DFA (see examples in Appendix \ref{app:E}). 

\begin{definition}[Canonical form of RLS] \label{def:canonical_form_RLS}
    \label{def:can_repr}
    The canonical MPS representation of $L_q$ thus corresponds to
    \begin{align} \label{eq:L^N_canonical_form_RLS}
        | L^N \rangle = \sum_m \sum_j
        \left(
        \begin{tikzpicture}[scale=.5, baseline={([yshift=-0.5ex]current bounding box.center)}, thick]
            \FullMPS{(0,0)}{\small $l_j^m$}{purple}
            \draw[fill=amaranth2!70] (-1.2,0) circle (0.2);
            \draw[fill=amaranth2!70] (6.2,0) circle (0.2);
            \FullMPS{(0,1.5)}{\small $\hat{x}_j^m$}{yellow}
            \draw[fill=amaranth] (-1.2,1.5) circle (0.2);
            \draw[fill=amaranth] (6.2,1.5) circle (0.2);
            \node at (3.2,2) {\scriptsize $N$ times};
        \end{tikzpicture}
        \right),
    \end{align}
    where the tensors $l_j^m, x_j^m$ are the canonical DFA tensors of $L_j^m, X_j^m$, and the MPO on top is defined as
    \begin{equation*}
        \begin{tikzpicture}[scale=.45, baseline={([yshift=0ex]current bounding box.center)}, thick]
            \MPSTensor{0,0}{\scriptsize $\hat{x}_j^m$}{yellow}
            \draw (0,-1) -- (0,-0.5);
        \end{tikzpicture}
        =  
        \begin{tikzpicture}[scale=.45, baseline={([yshift=-0.5ex]current bounding box.center)}, thick]
            \draw (-0.6,0) -- (0.6,0);
        \end{tikzpicture}
        \otimes
        \begin{tikzpicture}[scale=.45, baseline={([yshift=-0.5ex]current bounding box.center)}, thick]
            \draw (0,-1) -- (0,1);
            \filldraw[fill=purple] (-0.55,-1/2) -- (-0.55,1/2) -- (0.55,1/2) -- (0.55,-1/2) -- (-0.55,-1/2);
            \node at (0,0) {\scriptsize $\mathbb{P}_\infty$};
        \end{tikzpicture}
        + \begin{tikzpicture}[scale=.45, baseline={([yshift=-0.5ex]current bounding box.center)}, thick]
            \MPSTensor{0,0}{\scriptsize $x_j^m$}{yellow}
            \draw[fill=purple] (0,-1.15) circle (0.47);
            \node at (0,-1.15) {\scriptsize $\bra{f}$};
            \draw (0,-1.62) -- (0,-2);
        \end{tikzpicture},
    \end{equation*}
    where $\mathbb{P}_\infty := \sum_{i=0}^{|\Sigma_\infty|-1} \dyad{i}$.
\end{definition}
Therefore, the MPO effectively implements the action of $\hat{S}^{(m)}$, by letting the elements in $\Sigma_\infty$ stay the same, and substituting the $f$ symbols by the strings of $X_j^m$.

\marta{Let us illustrate Definition \ref{def:canonical_form_RLS} with our previous example $L_1 = 0^* 1 0^*$, which is canonically decomposed as $\{(L_1^1, X_1^1)\} \equiv \{(0^* f 0^*, 1)\}$, whose canonical DFAs are
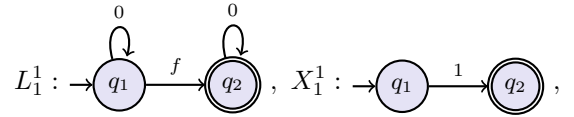
\begin{align*}
    L_1^1 : \hspace{-2mm}
    \begin{tikzpicture}[scale=.2, baseline={([yshift=-3ex]current bounding box.center)}, thick]
        \tikzstyle{small state} = [state, minimum size=0pt, fill=purple,node distance=1.5cm,initial distance=1.5cm,initial text=$ $]
        \tikzset{->}
        \node[small state, initial] (q1) {\small $q_1$};
        \node[small state, accepting, right of=q1] (q2) {\small $q_2$};
        \draw (q1) edge[loop above] node{\scriptsize $0$} (q1)
        (q2) edge[loop above] node{\scriptsize $0$} (q2)
        (q1) edge[above] node{\scriptsize $f$} (q2);
    \end{tikzpicture} \ ,
    \ 
    X_1^1 : \hspace{-2mm}\begin{tikzpicture}[scale=.2, baseline={([yshift=-0.30ex]current bounding box.center)}, thick]
        \tikzstyle{small state} = [state, minimum size=0pt, fill=purple,node distance=1.5cm,initial distance=1.5cm,initial text=$ $]
        \tikzset{->}
        \node[small state, initial] (q1) {\small $q_1$};
        \node[small state, accepting, right of=q1] (q2) {\small $q_2$};
        \draw (q1) edge[above] node{\scriptsize $1$} (q2);
    \end{tikzpicture} \ ,
    \ 
\end{align*}
respectively. Using the correspondence in Theorem \ref{lemma:NFAMPS}, this results in the MPS tensors
\begin{align*}
    \begin{tikzpicture}[scale=.45, baseline={([yshift=-1.5ex]current bounding box.center)}, thick]
        \MPSTensor{0,0}{\scriptsize $l_1^1$}{purple}
        \node at (0,1.3) {\scriptsize $0$};
    \end{tikzpicture}
    = 
    {\scriptsize \begin{pmatrix}
        1 & 0 \\ 0 & 1
    \end{pmatrix} }
    , \ 
    \begin{tikzpicture}[scale=.45, baseline={([yshift=-1.5ex]current bounding box.center)}, thick]
        \MPSTensor{0,0}{\scriptsize $l_1^1$}{purple}
        \node at (0,1.35) {\scriptsize $f$};
    \end{tikzpicture}
    = 
    {\scriptsize \begin{pmatrix}
        0 & 1 \\ 0 & 0
    \end{pmatrix} }
    , \ 
    \begin{tikzpicture}[scale=.45, baseline={([yshift=-1.5ex]current bounding box.center)}, thick]
        \MPSTensor{0,0}{\scriptsize $x_1^1$}{yellow}
        \node at (0,1.3) {\scriptsize $1$};
    \end{tikzpicture}
    = 
    {\scriptsize \begin{pmatrix}
        0 & 1 \\ 0 & 0
    \end{pmatrix} }
    .
\end{align*}
and the boundary vectors are $\bra{v_l} = (1,0)$ and $\ket{v_r} = (0,1)^T$ for both. 

Therefore, the final tensors appearing in the $\ket{L_N}$ expression of Eq. \eqref{eq:L^N_canonical_form_RLS} are given by  
\begin{equation*}
    \begin{tikzpicture}[scale=.45, baseline={([yshift=0ex]current bounding box.center)}, thick]
        \MPSTensor{0,0}{\scriptsize $\hat{x}_1^1$}{yellow}
        \draw (0,-1) -- (0,-0.5);
    \end{tikzpicture}
    =  
    \begin{tikzpicture}[scale=.45, baseline={([yshift=-0.5ex]current bounding box.center)}, thick]
        \draw (-0.6,0) -- (0.6,0);
    \end{tikzpicture}
    \otimes
    \begin{tikzpicture}[scale=.45, baseline={([yshift=-0.5ex]current bounding box.center)}, thick]
        \draw[fill=purple] (0,-0.5) circle (0.4);
        \node at (0,-0.5) {\scriptsize $0$};
        \draw (0,-0.9) -- (0,-1.4);
        \draw (0,0.9) -- (0,1.4);
        \draw[fill=purple] (0,0.5) circle (0.4);
        \node at (0,0.5) {\scriptsize $0$};
    \end{tikzpicture}
    + \begin{tikzpicture}[scale=.45, baseline={([yshift=-0.5ex]current bounding box.center)}, thick]
        \MPSTensor{0,0}{\scriptsize $x_1^1$}{yellow}
        \draw (0,-1.55) -- (0,-2.05);
        \draw[fill=purple] (0,-1.15) circle (0.4);
        \node at (0,-1.15) {\scriptsize $f$};
    \end{tikzpicture},
\end{equation*}
which has the following non-zero entries:
\begin{equation*}
    \begin{tikzpicture}[scale=.45, baseline={([yshift=-0.5ex]current bounding box.center)}, thick]
        \MPSTensor{0,0}{\scriptsize $\hat{x}_1^1$}{yellow}
        \draw (0,-1) -- (0,-0.5);
        \node at (0,1.3) {\scriptsize $0$};
        \node at (0,-1.3) {\scriptsize $0$};
    \end{tikzpicture}
    = 
    {\scriptsize \begin{pmatrix}
        1 & 0 \\ 0 & 1
    \end{pmatrix}}, 
    \ 
    \begin{tikzpicture}[scale=.45, baseline={([yshift=-0.5ex]current bounding box.center)}, thick]
        \MPSTensor{0,0}{\scriptsize $\hat{x}_1^1$}{yellow}
        \draw (0,-1) -- (0,-0.5);
        \node at (0,1.3) {\scriptsize $1$};
        \node at (0,-1.3) {\scriptsize $f$};
    \end{tikzpicture}
    = 
    {\scriptsize \begin{pmatrix}
        0 & 1 \\ 0 & 0
    \end{pmatrix}}. 
\end{equation*}
}




\section{Fundamental theorem of RLS}

Let us now address the equivalence of two RLS under local unitary operations, using the canonical form of RLS. To achieve this, we restrict to the class of \textit{sparse} regular languages, which have a number of words that scales polynomially with the word length, i.e. $\| | L_N \rangle \| = O(\text{poly}(N))$. These languages are well-characterized \cite{Szilard_1992} and widely used due to the fact that their reduced complexity can simplify certain problems \cite{Berman_1977, Hartmanis_1980, Szilard_1992, Dalessandro_2006_sparse_contextfree_languages, Hoffmann_2021, fernau_2021_finiteautomataintersectionnonemptiness, Hoffmann_2021}. It can also be efficiently decided whether a language is sparse or not \cite{gawrychowski_2010_sparse_efficiently}. \marta{Note that the physically relevant examples \textit{(a)-(e)} introduced in Section \ref{sec:RLS} are all sparse.}


\begin{restatable}[Fundamental theorem of sparse RLS]{theorem}{propLU}\label{prop:LU}
    Given two sparse RLS $L_1, L_2$ with canonical forms $\{(L_{1,j}^m, X_{1,j}^m)\}_{j,m}$ and $\{(L_{2,j}^m, X_{2,j}^m)\}_{j,m}$, they are related as
    \begin{equation*}
        \ket{L_2^N} = U^{\otimes N} \ket{L_1^N}, \ \forall N ,
    \end{equation*}
    for some unitary $U$ if and only if there exists a relabeling $\pi$ of the $\Sigma_\infty$ symbols and a unitary $U_f \in \mathcal{U}_{|\Sigma_f|}$ such that
    \begin{equation*} \begin{cases}
            L_{2,j}^{m} = \pi(L_{1,j}^{m}), \\
            |X_{2,j}^{(m)}\rangle := U_f^{\otimes m} | X_{1,j}^{m} \rangle,
        \end{cases}
        \forall j,m,
    \end{equation*}
    upon some reordering of the canonical form elements.
\end{restatable}

The above theorem demonstrates that LU equivalence for all $N$ can be tested in terms of a constant number of finite-size sub-problems. The proof is provided in Appendix \ref{app:F}.

\marta{As an illustration of this result, let us understand which languages are related to the W-state language $L_1 = 0^* 1 0^*$ through local unitary operations. Since its canonical decomposition is $\{(L_1^1, X_1^1)\} = \{(0^* f 0^*, 1)\}$, and thus $\Sigma = \{0\}$, $\Sigma_f = \{1\}$, the only possible relabelling is either the identity or the swap (i.e. $\pi(0) = 1$ and $\pi(1) = 0$), and the unitary is 1-dimensional so it trivializes and $\ket{X_{1,1}^1} = \ket{X_{2,1}^1}$, meaning that any other language $L_2$ that is locally unitarily equivalent to $L_1$ for all $N$ must necessarily be either $L_2 = 0^* 1 0^*$ or $1^* 0 1^*$.}

There is a significant difference between the fundamental theorems of TI MPS and RLS. In the standard case of TI MPS, it enables us to translate local unitaries acting at the physical level into invertible matrices at the bond level. Therefore, the bond dimension does not change, which reflects the fact that local changes of basis do not affect the entanglement content of the state. However, this is not necessarily the case for RLS, as illustrated by the following example. Consider the finite languages $L_1, L_2$ with canonical representations $\{(ff, 11 \cup 22 \cup 31 \cup 32)\}$ and $\{(ff, 11 \cup 12 \cup 32 \cup 33)\}$. Even though they are related as $\ket{L_2} = U^{\otimes 2} \ket{L_1}$ for a unitary $U$, their minimal DFA's have bond dimensions 5 and 4, as shown below. 
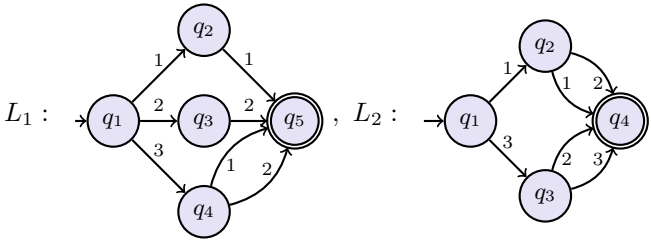
\begin{align*}
    L_1 : 
    \begin{tikzpicture}[scale=.1, baseline={([yshift=0ex]current bounding box.center)}, thick]
    \tikzstyle{small state} = [state, minimum size=0pt, fill=purple,node distance=1.2cm,initial distance=1.5cm,initial text=$ $]
    \tikzset{->}
        \node[small state, initial] (q1) {\small $q_1$};
        \node[small state, right of=q1] (q3) {\small $q_3$};
        \node[small state, above of=q3] (q2) {\small $q_2$};
        \node[small state, below of=q3] (q4) {\small $q_4$};
        \node[small state, right of=q3, accepting] (q5) {\small $q_5$};
        \draw
        (q1) edge[above] node{\scriptsize $1$} (q2)
        (q1) edge[above] node{\scriptsize $2$} (q3)
        (q1) edge[above] node{\scriptsize $3$} (q4)
        (q2) edge[above] node{\scriptsize $1$} (q5)
        (q3) edge[above] node{\scriptsize $2$} (q5)
        (q4) edge[bend left, below] node{\scriptsize $1$} (q5)
        (q4) edge[bend right, above] node{\scriptsize $2$} (q5);
\end{tikzpicture}\ , \ 
    L_2 :  \begin{tikzpicture}[scale=.175, baseline={([yshift=0ex]current bounding box.center)}, thick]
    \tikzstyle{small state} = [state, minimum size=0pt, fill=purple,node distance=1.4cm,initial distance=1.5cm,initial text=$ $]
    \tikzset{->}
        \node[small state, initial] (q1) {\small $q_1$};
        \node[small state, above right of=q1] (q2) {\small $q_2$};
        \node[small state, below right of=q1] (q3) {\small $q_3$};
        \node[small state, below right of=q2, accepting] (q4) {\small $q_4$};
        \draw
        (q1) edge[above] node{\scriptsize $1$} (q2)
        (q1) edge[above] node{\scriptsize $3$} (q3)
        (q2) edge[bend left, below] node{\scriptsize $2$} (q4)
        (q2) edge[bend right, above] node{\scriptsize $1$} (q4)
        (q3) edge[bend left, below] node{\scriptsize $2$} (q4)
        (q3) edge[bend right, above] node{\scriptsize $3$} (q4);
\end{tikzpicture} \ .
\end{align*}

Since the size of the minimal DFA is a fundamental measure that quantifies the descriptional complexity of RL's \cite{brzozowski_2017_complexityRLs, Yu_1994_state_complexity_RLoperations}, we see that, even though LU operations do not change the entanglement content of RLS, they can still modify them in a way that is not captured merely by the bond dimension of the TI PBC canonical form.

\section{Outlook}

We have delved into the connection between tensor networks and regular languages, exploring a potential cross-fertilization of ideas and tools between them. Indeed, we have seen how RL techniques can aid in studying physically relevant families of states and in answering open questions about MPS. In this regard, we characterized MPS representing RLS and TI RLS, and developed a canonical form for them, since the standard canonical form of MPS is not generally applicable to RLS. Then, we used it to study the LU-equivalence of sparse RLS, thus making a first step towards both a phase classification for RLS, and a generalized understanding of their LU- and SLOCC-equivalence, a key task that is generally difficult for arbitrary multipartite quantum states \cite{li_2018_slocc2}. Additional examples are provided in Appendix \ref{app:G}. 

One could also consider RLS with positive or complex coefficients. The resulting states would be linked to the areas of probabilistic graphical models and weighted finite automata (WFA), whose connection with tensor networks has been explored in previous works \cite{crosswhite_2008, orus_2017, Glasser_2019, li_2022_WFAMPS, riechers_2025_identifiabilityminimalityboundsquantum}. 
\marta{These, in turn, are commonly used in machine learning contexts: not only for probabilistic modeling, but also to improve the interpretability of black-box models and as the basis of learning algorithms for recurrent architectures \cite{li_2022_WFAMPS, okudono_weighted_2020, weiss_learning_2019, ayache_2019_wfa-explain-blackboxes, liu_2023_transformers-krohn-rhodes}. Therefore, it would be interesting to clarify the relationship between canonical forms developed in the WFA literature (e.g. the minimal form or the Krohn-Rodes decomposition \cite{Berstel_2010_noncommutative-rational-series, krohn_rhodes_1965}) and both our canonical decomposition and the standard MPS canonical form. A systematic understanding of these connections could enhance the interpretability of WFA-based representations and offer new tools to extend the LU-equivalence result of Theorem~\ref{prop:LU} beyond the sparse RLS setting and to SLOCC operations.} 

The next natural question that arises is whether there is a generalization of the RLS framework to two dimensions, or to broader types of languages such as the so-called context-free languages \cite{gopalakrishnan_2023_pushdownautomatasequentialgenerators, Zhang_2025_2D-super-area-law-states}. Regarding the two-dimensional setting, the concept of \textit{regular pictures} in 2D has been explored in the literature, aiming to retain the properties of 1D regular languages. Among the proposed extensions, the \textit{2D online tessellation automaton} (OTA) stands out for having an equivalent characterization in terms of 2D regular expressions, as well as other desirable properties \cite{Inoue_1977_DOTA_NOTA, Giammarresi_1997, Zhang_2025_2D-super-area-law-states}. OTA, in turn, can be formulated as PEPS with tensors of a special, simple form (see Appendix \ref{app:H}). However, unlike 1D RLS, there is no efficient characterization of the class of PEPS representing 2D RLS, nor a canonical form, since these problems are undecidable even when restricted to the very special form of the OTA PEPS tensors \cite{Inoue_1977_DOTA_NOTA, Anselmo_2006_UOTA}.

\marta{This work lays the foundation for extending MPS theory beyond RLS. Following on this, \cite{florido-llinas_2025_MPS-X} introduces a generalized canonical form for uniform MPS with an arbitrary boundary matrix, in which any such MPS is decomposed in terms of a backbone family of states called \textit{algebraic RLS} that encodes its long-range, scale-invariant information. Algebraic RLS slightly generalize a subclass of the RLS introduced in this work, and inherit both their structural properties and the canonical decomposition developed here. This could enable extending results in the quantum many-body literature that are currently restricted to standard MPS or to highly specific families such as W- or Dicke-type MPS (e.g. the parent Hamiltonian construction for ground states \cite{garre-rubio_2025_mpsstabilityintersectionproperty} or for eigenstates that arise as quantum many-body scars \cite{gioia_2025_distincttypesparenthamiltonians}).}


\marta{Moreover, identifying RLS as a structured and physically relevant class of quantum states, and formulating a theoretical framework for them, allows designing a compiler that explicitly exploits their automaton-like structure to yield more efficient quantum circuits for the preparation of RLS than current general-purpose or sparse-state compilers \cite{bellante_2025_compiling-RLS}. Since state preparation is a central subroutine in many quantum algorithms, improvements in this step can have a significant practical impact. 

Beyond circuit synthesis, it remains an open direction for future work to determine whether specific families of RLS may offer advantages in certain tasks, in analogy with how GHZ states and their generalizations (which are RLS themselves) are exploited in many quantum information protocols and in quantum sensing \cite{Pan_2012_GHZ, erhard_2020_advances-ghz, pezze_2018_qmetrology-ghz}.}

\section*{Acknowledgements}

MFL thanks Georgios Styliaris for helpful discussions and acknowledges support from the International Max-Planck Research School for Quantum Science and Technology (IMPRS-QST). This research is part of the Munich Quantum Valley (MQV), which is supported by the Bavarian State Government with funds from the High-tech Agenda Bayern Plus. AMA acknowledges support from the Spanish Agencia Estatal de Investigacion through the grants ``IFT Centro de Excelencia Severo Ochoa SEV-2016-0597" and ``Ram\'on y Cajal RyC2021-031610-I'', financed by MCIN/AEI/10.13039/501100011033 and the European Union NextGenerationEU/PRTR. DPG acknowledges support from the Spanish Ministry of Science and Innovation MCIN/AEI/10.13039/501100011033 (grants CEX2023-001347-S and PID2020-113523GB-I00). This work has been financially supported by the Ministry for Digital Transformation and of Civil Service of the Spanish Government through the QUANTUM ENIA project call - Quantum Spain project, and by the European Union through the Recovery, Transformation and Resilience Plan – NextGenerationEU within the framework of the Digital Spain 2026 Agenda.

\input{appendices_no_colors}

\bibliography{references.bib}

\end{document}

%% file: packages.tex
\usepackage{graphicx}
\usepackage{dcolumn}
\usepackage{bm}

\usepackage{dsfont} 
\usepackage{physics} 
\usepackage{blindtext}
\usepackage{amsfonts, amsthm, microtype, mathrsfs, bbm}
\usepackage[colorlinks,allcolors=blue]{hyperref}
\usepackage{cleveref}
\usepackage{mathtools, enumerate}
\usepackage{thmtools, thm-restate}
\usepackage[dvipsnames,usenames]{xcolor}

\usepackage[ruled]{algorithm2e}

\usepackage{csquotes} 

\newtheorem{definition}{Definition}
\newtheorem*{definitionnon}{Definition}
\newtheorem{theorem}{Theorem}

\newtheorem{corollary}[theorem]{Corollary}
\theoremstyle{definition}

\usepackage{tikz}
\usepackage{venndiagram}
\usetikzlibrary{decorations.pathreplacing, calligraphy, decorations.markings}
\usetikzlibrary{arrows, automata, positioning}

\definecolor{purple}{HTML}{E5E3F5}
\definecolor{amaranth}{HTML}{F5CFD8}
\definecolor{yellow}{HTML}{FFEAB6}
\definecolor{pink}{HTML}{FFDBDA}

\definecolor{darkpurple}{HTML}{3D348B}
\definecolor{darkamaranth}{HTML}{FF1A53}
\definecolor{amaranth2}{HTML}{da3e52}
\definecolor{darkyellow}{HTML}{FFBA08}
\definecolor{darkpink}{HTML}{FFDBDA}

\newcommand{\MPSTensor}[3]{
	\begin{scope}[shift={(#1)}]
		\draw (-1,0) -- (1,0);
		\draw (0,1) -- (0,0);
		\filldraw[fill=#3] (-1/2,-1/2) -- (-1/2,1/2) -- (1/2,1/2) -- (1/2,-1/2) -- (-1/2,-1/2);
		\draw (0,0) node {\scriptsize #2};
	\end{scope}
}

\newcommand{\FullMPS}[3]{
	\begin{scope}[shift={(#1)}]
		\draw[shift={(0,0)},dotted] (0,0) -- (4.5,0);
        \MPSTensor{0,0}{#2}{#3}
        \MPSTensor{1.5,0}{#2}{#3}
        \MPSTensor{5,0}{#2}{#3}
	\end{scope}
}

\newcommand{\MatrixHorizontalIndices}[4]{
	\begin{scope}[shift={(#1)}]
		\draw (-1-#4/2,0) -- (1+#4/2,0);
		\filldraw[fill=#3] (-1/2-#4/2,-1/2) -- (-1/2-#4/2,1/2) -- (1/2+#4/2,1/2) -- (1/2+#4/2,-1/2) -- (-1/2-#4/2,-1/2);
		\draw (0,0) node {\scriptsize #2};
	\end{scope}
}

\newcommand{\pepstensor}[8]{
	\begin{scope}[shift={(#1)}]
		\draw (-2,0) -- (2,0);
        \draw (0,-2) -- (0,2);
		\filldraw[fill=#3] (-1,-1) -- (-1,1) -- (1,1) -- (1,-1) -- (-1,-1);
		\draw (0,0) node {\footnotesize #2};
        \draw[line width=1.25] (0,0) -- (1.5,1.5);
        \node at (0,2.5) {\scriptsize #4};
        \node at (2.5,0) {\scriptsize #5};
        \node at (0,-2.5) {\scriptsize #6};
        \node at (-2.5,0) {\scriptsize #7};
        \node at (2,2) {\scriptsize #8};
	\end{scope}
}

\newcommand{\pepstensorabajo}[8]{
	\begin{scope}[shift={(#1)}]
		\draw (-2,0) -- (2,0);
        \draw (0,0) -- (0,2);
		\filldraw[fill=#3] (-1,-1) -- (-1,1) -- (1,1) -- (1,-1) -- (-1,-1);
		\draw (0,0) node {\footnotesize #2};
        \draw[line width=1.25] (0,0) -- (1.5,1.5);
        \node at (0,2.5) {\scriptsize #4};
        \node at (2.5,0) {\scriptsize #5};
        \node at (0,-2.5) {\scriptsize #6};
        \node at (-2.5,0) {\scriptsize #7};
        \node at (2,2) {\scriptsize #8};
	\end{scope}
}

\newcommand{\pepstensorderecha}[8]{
	\begin{scope}[shift={(#1)}]
		\draw (-2,0) -- (0,0);
        \draw (0,0) -- (0,2);
		\filldraw[fill=#3] (-1,-1) -- (-1,1) -- (1,1) -- (1,-1) -- (-1,-1);
		\draw (0,0) node {\footnotesize #2};
        \draw[line width=1.25] (0,0) -- (1.5,1.5);
        \node at (0,2.5) {\scriptsize #4};
        \node at (2.5,0) {\scriptsize #5};
        \node at (0,-2.5) {\scriptsize #6};
        \node at (-2.5,0) {\scriptsize #7};
        \node at (2,2) {\scriptsize #8};
	\end{scope}
}

\newcommand{\fullpepstensor}[1]{
    \begin{scope}[shift={(#1)}]
        \pepstensorabajo{0,0}{}{purple}{}{}{}{}{}
        \pepstensorabajo{3,0}{}{purple}{}{}{}{}{}
        \pepstensor{9,0}{}{purple}{}{}{}{}{}
        \draw[dotted] (5,0) -- (7,0);

        \pepstensor{0,6}{}{purple}{}{}{}{}{}
        \pepstensor{3,6}{}{purple}{}{}{}{}{}
        \pepstensorderecha{9,6}{}{purple}{}{}{}{}{}
        \draw[dotted] (5,6) -- (7,6);

        \pepstensor{0,9}{}{purple}{}{}{}{}{}
        \pepstensor{3,9}{}{purple}{}{}{}{}{}
        \pepstensorderecha{9,9}{}{purple}{}{}{}{}{}
        \draw[dotted] (5,9) -- (7,9);

        \draw[fill=amaranth] (0,11.75) circle (0.75);
        \draw[fill=amaranth] (3,11.75) circle (0.75);
        \draw[fill=amaranth] (9,11.75) circle (0.75);

        \draw[fill=amaranth] (-2.75,0) circle (0.75);
        \draw[fill=amaranth] (-2.75,6) circle (0.75);
        \draw[fill=amaranth] (-2.75,9) circle (0.75);

        \draw[fill=yellow] (11.75,0) circle (0.75);

        \draw[fill=yellow] (9,-2.75) circle (0.75);

        \draw[dotted] (0,2) -- (0,4)
        (3,2) -- (3,4)
        (9,2) -- (9,4);
	\end{scope}
}

\newcommand{\boundarypeps}[2]{
    \begin{scope}[shift={(#1)}]
        \draw[fill=#2] (0,0) circle (0.75);
        \draw (-0.75,0) -- (-2,0);
    \end{scope}
}

\newcommand{\FullMPSX}[5]{
	\begin{scope}[shift={(#1)}]
		\draw[shift={(0,0)},dotted] (0,0) -- (4.5,0);
        \MPSTensor{0,0}{#2}{#4}
        \MPSTensor{1.5,0}{#2}{#4}
        \MPSTensor{4.5,0}{#2}{#4}
        \MatrixHorizontalIndices{-1.5,0}{#3}{#5}{0}
        \draw (-2.5,0) -- (-2.5,-0.8) -- (5.5,-0.8) -- (5.5,0);
	\end{scope}
}

%% file: appendices_no_colors.tex
\appendix

\section{Proof of Theorem \ref{lemma:NFAMPS} connecting RLS and MPS} \label{app:A}

\lemmaNFAMPS*

\begin{proof}
    Kleene's theorem asserts that every RL can be equivalently defined either through a regular expression or a finite automaton \cite{Kleene_1956}. Then, given a RL $L$ and a finite automaton that accepts it, $\mathcal{F} = \langle Q, \Sigma, \delta, I, F \rangle$, if we define tensor $A$ and vectors $v_l, v_r$ according to Eq. \eqref{eq:associate_MPS_to_NFA1}, 
    then the language $L$ is equal to 
    \begin{equation*}
        L = \big\{ x_1 x_2 \dots x_N \in \Sigma^* \mid 
        \begin{tikzpicture}[scale=.45, baseline={([yshift=-2.5ex]current bounding box.center)}, thick]
            \FullMPS{0,0}{$A$}{purple}
            \draw[fill=amaranth] (-1.3,0) circle (0.4);
            \node at (-1.3,0) {\scriptsize $v_l$};
            \draw[fill=amaranth] (6.3,0) circle (0.4);
            \node at (6.3,0) {\scriptsize $v_r$};
            \draw[fill=yellow] (0,1.3) circle (0.4);
            \node at (0,1.3) {\scriptsize $x_1$};
            \draw[fill=yellow] (1.5,1.3) circle (0.4);
            \node at (1.5,1.3) {\scriptsize $x_2$};
            \draw[fill=yellow] (5,1.3) circle (0.4);
            \node at (5,1.3) {\scriptsize $x_N$};
        \end{tikzpicture}
        > 0
        \big\},
    \end{equation*}
    and the states $\ket{L_N}$ defined in Eq. \eqref{eq:familyRLMPS1} are
    \begin{equation*}
        \ket{L_N} = \sum_{w \in L \cap \Sigma^N} c_w \ket{w},
    \end{equation*}
    where $c_w$ equals the number of accepting paths for word $w$ through the automaton. Therefore, when $\mathcal{F}$ is an UFA, all of the coefficients $c_w$ are one, and thus the family of states $\{\ket{L_N}\}_N$ is a RLS.  
\end{proof}

{}

\section{Proof of Lemma \ref{lemma:UFAMPS} to determine which MPS are RLS} \label{app:UFAMPS} \label{app:B}

\UFAMPS*

\begin{proof}
    The proof relies on Proposition 1.15 in \cite{Sakarovitch2009}. First, we need some preliminary definitions.

    An automaton $\mathcal{F}$ is said to be \textit{trim}, if all of its states are \textit{useful}. This means that each state is accessible from at least one of the initial states, and at least one path that reaches one of the accepting states arises from it. 
    
    Equivalently, state $n$ is useful if there exist two strings $v, w \in \Sigma^*$ such that
    \begin{equation} \label{eq:def_useful}
        \mel{v_l}{A^v}{n} \neq 0 \text{ and } \mel{n}{A^w}{v_r} \neq 0. 
    \end{equation}
    where tensor $A$ and vectors $v_l, v_r$ are obtained from $\mathcal{F}$ according to Theorem \ref{lemma:NFAMPS}, and $A^y := A^{y_1} A^{y_2} \dots A^{y_l}$ for any string $y := y_1 y_2 \dots y_l \in \Sigma^*$.

    Therefore, $n$ is useless if and only if
    \begin{equation} \label{eq:useless_cond}
        \ket{n} \in (\mathcal{A}^T \ket{v_l})^\perp \cup (\mathcal{A} \ket{v_r})^\perp ,
    \end{equation}
    where $\mathcal{A} = \text{Alg}(\{A^x\})$. To assert this, we used the facts that $\mathcal{A}$ admits a basis $\{A^{w_1}, \dots, A^{w_M}\}$ in terms of a set of words $w_1, \dots, w_M \in \Sigma^*$ (it is enough to consider $w_i$ of length upper bounded by $D^2 - 1$ \cite{Pappacena_1997_length-algebra}), and also $\mathcal{A}^\dagger = \mathcal{A}^T$ because there are no complex entries.

    On the other hand, the \textit{product automaton} of $\mathcal{F}$, denoted as $\mathcal{F} = \langle Q \times Q, \Sigma, \delta', I \times I, F \times F \rangle$, is the Cartesian product of $\mathcal{F}$ with itself, where 
    \begin{equation*}
        (k, l) \in \delta'((i,j), x) \iff k \in \delta(i, x) \text{ and } l \in \delta(j, x).
    \end{equation*}
    The MPS tensors $A', v_l', v_r'$ associated to $\mathcal{F} \times \mathcal{F}$ are thus $(A')^x = A^x \otimes A^x$, $\bra{v'_l} = \bra{v_l} \otimes \bra{v_l}$ and $\ket{v'_r} = \ket{v_r} \otimes \ket{v_r}$. The \textit{diagonal part} of $\mathcal{F} \times \mathcal{F}$ is the sub-automaton that arises from restricting to the set of internal states $\{(i,i)\}_{i \in Q}$. 
    
    By defining the vector subspaces $\mathcal{V}_L := \{( a^T \ket{v_l} )^{\otimes 2} \mid a \in \mathcal{A} \}$ and $\mathcal{V}_R := \{(a \ket{v_r})^{\otimes 2} \mid a \in \mathcal{A}\}$, Eq. \eqref{eq:useless_cond} tells us that state $(m, n)$ in $\mathcal{F} \times \mathcal{F}$ being useless is equivalent to
    \begin{equation} \label{eq:lemma2ufa}
        \ket{m} \otimes \ket{n} \in (\mathcal{V}_L)^\perp \cup (\mathcal{V}_R)^\perp.
    \end{equation}
    To prove the first part of the lemma, we can now use Proposition 1.15 in \cite{Sakarovitch2009}, which states that a finite automaton $\mathcal{F}$ is unambiguous if and only if the trim part of $\mathcal{F} \times \mathcal{F}$ is equal to its diagonal part, meaning that all states $(m,n) \in Q \times Q$ with $m \neq n$ are \textit{useless}. Therefore, the MPS is a RLS if and only if Eq. \eqref{eq:lemma2ufa} is true for all $m \neq n$.

    Finally, we show that this condition can be checked in time $O(dD^3)$ in Algorithm \ref{algorithm:check_MPS_RLS} box. The procedure we propose scales better than the standard way to check unambiguity according to \cite{Sakarovitch2009}, since it does not require the construction of the product automaton. The latter runs in time $O(m^2)$, where $m$ is the number of transitions in the NFA, which scales at worst as $O(d^2 D^4)$ since $m \leq dD^2$.

    \begin{algorithm}[h!] 
      \SetAlgoLined
      \caption{Check if an MPS is a RLS} \label{algorithm:check_MPS_RLS}
      \BlankLine
      \small
      \KwData{NFA $\mathcal{F} = \langle Q, \Sigma, \delta, I, F \rangle$ with $|Q| = D$, $|\Sigma| = d$.}
      \KwResult{Decide if the MPS associated to $\mathcal{F}$ is a RLS in time $O(dD^3)$.}
      \BlankLine
      \Begin
      {    
        $\{\ket{l_i}\}_i \gets$ Basis of $\mathcal{A}^T \ket{v_l}$ \;
        
        $\{\ket{r_i}\}_i \gets$ Basis of $\mathcal{A} \ket{v_r}$ \;

        \tcp{These bases can be found in time $O(dD^3)$ \cite{Kiefer_2011}.}
        
        \ForEach{$i \in \{1, \dots, D\}, \ket{l_j}, \ket{r_j}$}{
            $L_{ij} \gets \braket{i}{l_j}$ \;
            $R_{ij} \gets \braket{i}{r_j}$ \;
        }
        
        \tcp{Runtime of the block above: \hspace{-2mm}$O(D^3)$)}
        
        \ForEach{$m \neq n$, $m,n \in \{1, \dots, D\}$}{
            \lIf{$L_{mi} L_{ni} = 0$, $\forall i$}{$\ket{m} \otimes \ket{n} \in (\mathcal{V}_L)^\perp$} 

            \lIf{$R_{mi} R_{ni} = 0$, $\forall i$}{$\ket{m} \otimes \ket{n} \in (\mathcal{V}_R)^\perp$}

            \If{$\ket{m} \otimes \ket{n} \notin (\mathcal{V}_L)^\perp  \cup (\mathcal{V}_R)^\perp$}{
                \Return The MPS is not a RLS \;
            }
        }

        \tcp{Runtime of the block above: \hspace{-2mm}$O(D^3)$}
        
        \Return The MPS is a RLS \;
    }     
    \end{algorithm} 
\end{proof}

\section{Proof of Lemma \ref{lemma:TI} and Corollary \ref{corollary:TI} about translationally invariant RLS} \label{app:C}

\lemmaTI*

\begin{proof}
    First, let us reformulate the TI condition as follows. Given the state in Eq. \eqref{eq:general_MPS-X}
    \begin{equation*}
        \ket{\psi_N(X,A)} = \sum_{i_1, \dots, i_N} \Tr[X A^{i_1} A^{i_2} \dots A^{i_N}] \ket{i_1 i_2 \dots i_N},
    \end{equation*}
    it is TI for all $N$, if and only if, for all $r \leq N-1$, and for all values of $i_1, \dots, i_N \in \{0, \dots, d-1\}$, it holds that
    \begin{align*}
        \Tr[X A^{i_1} A^{i_2} \dots A^{i_N}] = \Tr[X A^{i_{r+1}} \dots A^{i_N} A^{i_1} \dots A^{i_r}]
    \end{align*}
    Regrouping the terms together, this is equivalent to
    \begin{equation} \label{eq:TI1}
        \Tr[X[ A^{i_1} \dots A^{i_r}, A^{i_{r+1}} \dots A^{i_N}]] = 0, \ \forall i_1, \dots, i_N.
    \end{equation}

    On the other hand, the algebra of the MPS matrices, $\mathcal{A} := \text{Alg}(\{A^i\})$, admits a basis of the form $\{A^{w_1}, \dots, A^{w_M}\}$, where $A^{w_i} := A^{w_{i,1}} \dots A^{w_{i,l_i}}$ and we can take $l_i \leq D^2 - 1$ \cite{Pappacena_1997_length-algebra}.

    Then, for any elements $a, b \in \mathcal{A}$, we know that there exist $\alpha_i, \beta_j \in \mathbb{C}$ such that
    \begin{equation*}
        a = \sum_i \alpha_i A^{w_i}, \ b = \sum_j \beta_j A^{w_j}. 
    \end{equation*}
    Therefore, $\ket{\psi_N(X,A)}$ are TI for all $N$ if and only if
    \begin{align*}
        &\Tr[X[a,b]] = \sum_i \sum_j \alpha_i \beta_j \Tr[X[A^{w_i}, A^{w_j}]] \\
        &= \sum_i \sum_j \alpha_i \beta_j \underbrace{\Tr[X[A^{w_{i,1}} \dots A^{w_{i,l_i}}, A^{w_{j,1}} \dots A^{w_{j,l_j}}]]}_{= 0 \text{ by Eq. } \eqref{eq:TI1}} \\
        &= 0, \quad \forall a, b \in \mathcal{A}.
    \end{align*}

\end{proof}

\corollaryTI*

\begin{proof}
    The first part of the statement is a direct application of Lemma \ref{lemma:TI}. To prove the runtime, we outline a scheme to check shift invariance in Algorithm \ref{alg:check_TI} box.

    \begin{algorithm}[h!] 
      \SetAlgoLined
      \caption{Decide if a RL is shift-invariant} \label{alg:decide_UFA}
      \label{alg:check_TI}
      \BlankLine
      \small
      \KwData{RL $L$ and any MPS representation $\{A, v_l, v_r\}$, with bond dimension $D$ and $|\Sigma| = d$.}
      \KwResult{Decide if $L$ is shift-invariant in time $O(dD^3)$.}
      \BlankLine
      \Begin
      {
        $\mathcal{S} := \{(w_i, \bra{l_i}, \ket{r_i})\} \gets$ Set of triples consisting of strings $w_i \in \Sigma^*$ and vectors $\bra{l_i} := \bra{v_l} A^{w_i}$, $\ket{r_i} := A^{w_i} \ket{v_r}$, which fully span $\bra{v_l} \mathcal{A}$ and $\mathcal{A} \ket{v_r}$, respectively. 
        
        \tcp{One can compute the set $\mathcal{S}$ by finding bases of $\bra{v_l} \mathcal{A}$ and $\mathcal{A} \ket{v_r}$, which is doable in time $O(dD^3)$ \cite{Kiefer_2011}. Note that $\mathcal{S}$ contains at most $2D$ elements.}
        
        \ForEach{$(w_i, \bra{l_i}, \ket{r_i}), (w_j, \bra{l_j}, \ket{r_j}) \in \mathcal{S}$}{
            \lIf{$\braket{l_i}{r_j} \neq \braket{l_j}{r_i}$}{\Return $L$ is not shift-invariant}
        }
        \tcp{Runtime of the block above: \hspace{-2mm}$O(D^3)$.}
        
        \Return $L$ is shift-invariant \;
    }
    \end{algorithm} 
\end{proof}

Note that a naive automata implementation to determine if a RL $L$ is shift-invariant can result in an algorithm scaling much worse with the bond dimension, $O(dD^2 4^{D^2})$. We show this in Algorithm \ref{alg:check_TI_naive} box. The improvement of our scheme relies on the fact that, while the construction of an equivalent DFA leads to an exponential overhead in the naive automata implementation, we can avoid doing so by leveraging Lemma \ref{lemma:TI}.

\begin{algorithm}[h!] 
      \SetAlgoLined
      \caption{Decide if a RL is shift-invariant (naive automata implementation)}
      \label{alg:check_TI_naive}
      \BlankLine
      \small
      \KwData{RL $L$ and any NFA representation with $D$ internal states and $|\Sigma| = d$.}
      \KwResult{Decide if $L$ is shift-invariant in time $O(d D^2 4^{D^2})$.}
      \BlankLine
      \Begin
      {
        $L' \gets \text{shift}(L) := \{vu \mid uv \in L\}$, through an NFA with at most $2D^2 + 1$ internal states \cite{Jiraskova_Okhotin_2008} \tcp*[l]{Runtime: \hspace{-2mm}$O(D^2)$}

        Construct an equivalent DFA accepting $L \gets$ Can be done with Rabin-Scott powerset construction, resulting in a DFA with at most $2^D$ states \cite{rabin_1959_powerset} \tcp*[l]{Runtime: \hspace{-2mm}$O(d D 2^D)$.}

        Construct an equivalent DFA accepting $L' \gets$ Results in at most $2^{2D^2+1}$ states \tcp*[l]{Runtime: \hspace{-2mm}$O(d D^2 4^{D^2})$.}

        \eIf{$L = L'$}{
            \Return $L$ is shift-invariant \;
        }{
            \Return $L$ is not shift-invariant \;
        }
        \tcp{Equivalence checking of these DFA's can be done in time $O(d D^2 4^{D^2})$ \cite{hopcroft_1971_linear_equiv}.}
    }
    \end{algorithm}

\section{Proof of Theorem \ref{prop:canonical_decomposition_RL} about the existence and uniqueness of a canonical form for RL's} \label{app:D}

\marta{
We first recall the notion of tree width for finite automata. A partial computation of a word $w$ on an automaton $\mathcal F$ is a path that reads $w$ for as long as transitions are defined, terminating once no outgoing transition exists. We denote by $\mathrm{tw}(\mathcal F,w)$ the total number of such partial computations.
\begin{definition}[Tree width]
    The \emph{tree width} of automaton $\mathcal{F}$ is defined as 
    \begin{equation*}
        \mathrm{TW}(\mathcal{F}) := \sup_{w \in \Sigma^*} \mathrm{tw}(\mathcal{F},w) \, .
    \end{equation*}
\end{definition}
By definition, any DFA has tree width 1. In contrast, even unambiguous NFAs may have infinite tree width. A simple example is the following automaton, 
\begin{equation} \label{eq:nondeterm_transition-cycle}
    \mathcal{F}: \begin{tikzpicture}[scale=.2, baseline={([yshift=-3ex]current bounding box.center)}, thick]
        \tikzstyle{small state} = [state, minimum size=0pt, fill=purple,node distance=1.5cm,initial distance=1.5cm,initial text=$ $]
        \tikzset{->}

        \node[small state, initial, accepting] (q1) {\small $q_1$};
        \node[small state, right of=q1] (q2) {\small $q_2$};

        \draw (q1) edge[loop above] node{\scriptsize $0,1$} (q1)
        (q1) edge[bend left, above] node{\scriptsize $0$} (q2);
    \end{tikzpicture} \ ,
\end{equation}
which provides an unambiguous representation of the language $(0\cup1)^*$. The presence of a nondeterministic transition on a cycle implies infinite tree width \cite{Palioudakis_2012_tree-width}: for any $\ell \in \mathbb N$, the word $w=0^\ell$ admits $\ell+1$ partial computations, namely the accepting path $q_1^\ell$ and the incomplete paths $q_1^{i}q_2$ for $i=0,\dots,\ell-1$.

The relevance of tree width in what follows stems from its role in bounding the cost of determinization. The standard Rabin–Scott powerset construction~\cite{rabin_1959_powerset} transforms an NFA with $D$ states into an equivalent DFA with at most $2^D$ states, in time $O(dD2^D)$ in the worst case. However, when the NFA has finite tree width~$k$, the resulting DFA has only $O(D^k)$ states and can be constructed in time $O(dD^{k+1})$~\cite{Palioudakis_2012_tree-width}. This distinction underlies the polynomial-versus-exponential complexity behavior in the canonical decomposition algorithm.}

\propcanRL*
\begin{proof}
    To begin with, we see that a decomposition with such properties is unique. Indeed, assume that $L$ admits two such decompositions, $\{(L_{1,j}^m, X_{1,j}^m)\}_{j,m}$ and $\{(L_{2,j}^m, X_{2,j}^m)\}_{j,m}$. 
    
    First, we see that $L_{1,j}^{m} = L_{2,j}^{m}$ upon some relabelling of the elements. Suppose that there exist $v, w \in L_{2,k}^m$ such that $v \in L_{1,i}^m$ and $w \in L_{1,j}^m$, for some $i \neq j$. This would imply that $X_{1,i}^m = X_{2,k}^m = X_{1,j}^m$, which is not possible unless $i = j$. Due to the symmetry of the argument under exchange of $L_1, L_2$, it follows that $L_{1,j}^{m} = L_{2,j}^{m}$ upon relabelling, and as a consequence, $X_{1,j}^m = X_{2,j}^m$.

    Now, we show that it is always possible to find the canonical decomposition of $L$ by providing an explicit algorithm to do so. We assume that we are given an NFA representation of $L$ of bond dimension $D$. 

    First of all, we find for each $m$ all finite words $x_1 \dots x_m \in \Sigma_f^m$ such that the languages
    \begin{equation} \label{eq:def_Lx1..xm}
        L_{x_1 \dots x_m} := T_{\Sigma_f \to f} (L \cap \Sigma_\infty^* x_1 \Sigma_\infty^* x_2 \dots \Sigma_\infty^* x_m \Sigma_\infty^*)
    \end{equation}
    are non-empty, where $T_{\Sigma_f \to f}$ replaces all symbols in $\Sigma_f$ by the single symbol $f$. The intersection in Eq. \eqref{eq:def_Lx1..xm} can be obtained through the product automaton construction \cite{Sakarovitch2009}, resulting in a bond dimension of at most $(m+1)D$, and computable in time $O(t_1 t_2)$, where $t_i$ is the number of edges in the NFA. Note that $t_1 \leq dD^2$ for the NFA accepting $L$, and $t_2 = (m+1)|\Sigma_\infty|+m$ \cite{Sakarovitch2009}.

    Moreover, the substitution $T_{\Sigma_f \to f}$ can also be done in time $O(t_1 t_2)$, since we just need to go through the $t_1 t_2$ edges of the product automaton and change the ones with label in $\Sigma_f$ by the output symbol $f$. 
    
    Therefore, the language can be expressed as 
    \begin{equation*} \small
        L = \cup_m \cup_{x_1 \dots x_m} S^{(m)} ( L_{x_1 x_2 \dots x_m}, x_1 x_2 \dots x_m ).
    \end{equation*}
    However, this decomposition does not satisfy the desired properties yet, since $\{L_{x_1 \dots x_m}\}$ are not necessarily pairwise disjoint. For example, for the language $L = 0^* 1 0^* 2 0^* \cup 0^* 2 0^* 1 0^*$, we would have that $L_{12} = L_{21} = 0^* f 0^* f 0^*$. We proceed to define a new partition $\{L_j^m\}$ of pairwise disjoint RL's, with the desirable property that $L_j^m \cap L_{x_1 \dots x_m} \neq \emptyset \implies L_j^m \subseteq L_{x_1 \dots x_m}$. 

    For each $m$, let $\mathcal{S} := \{S_1, S_2, \dots\}$ be the power set of the words $\{x_1 \dots x_m \in \Sigma_f^m \mid L_{x_1 \dots x_m} \neq \emptyset\}$. Then, the sets defined as
    \begin{equation} \label{eq:sets_Li_intersection}
        L_i^m := \left( \cap_{x_1 \dots x_m \in S_i} L_{x_1 \dots x_m} \right) \cap \left( \cap_{x_1 \dots x_m \notin S_i} L_{x_1 \dots x_m}^c  \right),
    \end{equation}
    where $L^c$ denotes the complement of $L$, achieve the desired properties, and can be computed again using automata tools \cite{hopcroft_1971_linear_equiv}. We neglect the $L_i^m$ that are empty. For clarity, an example of this partition appears in Figure \ref{fig:example_partition}.

    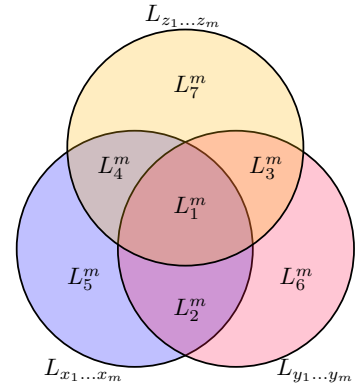
\begin{figure}[h]
    \centering
    \resizebox{0.6\columnwidth}{!}{%
    \begin{tikzpicture}[scale=0.3, baseline={([yshift=0ex]current bounding box.center)}, thick]
    \tikzset{venn circle/.style={draw=black,text opacity=1,fill opacity=0.25,circle,minimum width=3.5cm,fill=#1,line width=0.75}}
    \tikzset{label/.style={text width=1.5cm}}
    \begin{scope}
          \node [venn circle = blue] (A) at (0,0) {};
          \node [label] (A1) at (-2,-6) {$L_{x_1 \dots x_m}$};
          \node [venn circle = darkamaranth] (B) at (5,0) {};
          \node [label] (B1) at (9.5,-6) {$L_{y_1 \dots y_m}$};  
          \node [venn circle = darkyellow] (C) at (2.5,5) {};
          \node [label] (C1) at (3,11.5) {$L_{z_1 \dots z_m}$};

        \node at (2.75,2) {$L_1^m$};
        \node at (2.75,-3) {$L_2^m$};
        \node at (6.5,4) {$L_3^m$};
        \node at (-1,4) {$L_4^m$};
        \node at (-2.5,-1.5) {$L_5^m$};
        \node at (8,-1.5) {$L_6^m$};
        \node at (2.75,8) {$L_7^m$};
        \end{scope}
    \end{tikzpicture}}
    \caption{Example of how the partition in Eq. \eqref{eq:sets_Li_intersection} looks like in the proof of Theorem \ref{prop:canonical_decomposition_RL}, given three initial sets $L_{x_1 \dots x_m}, L_{y_1 \dots y_m}, L_{z_1 \dots z_m}$.}
    \label{fig:example_partition}
\end{figure}
    
    Then, we can rewrite the language as
    \begin{equation*} 
        L = \cup_m \cup_j S^{(m)} ( L_j^m , X_j^m ), 
    \end{equation*}
    where $X_j^m := \cup_{x_1 \dots x_m \in S_j} \ x_1 \dots x_m$. Note that, by construction, $X_j^m \neq X_k^m$ unless $j = k$. Therefore, the desired conditions on $L_j^m, X_j^m$ are satisfied. 

    \marta{We show in Algorithm \ref{alg:can_decomp} that this scheme to find the canonical decomposition has a worst-case complexity of $O(M(d2^{1+2(M+1)D})^{d^M})$ when the input automaton is an NFA of infinite tree width, which scales with the bond dimension $D$ as $O(4^{(m+1) d^M D})$. The exponential dependence on $D$ originates entirely from the determinization step in Algorithm \ref{alg:can_decomp}, which is required to construct the complements $L_{x_1 \dots x_m}^c$.
    
    We now show that, when the input automaton accepting $L$ has finite tree width $k$, the overall scaling becomes polynomial in $D$, rather than exponential. The central observation is that the tree width of the automata recognizing $L_{x_1 \dots x_m}$ (defined in Eq. \eqref{eq:def_Lx1..xm}) can be uniformly bounded by $k$, for each $x_1, \dots, x_m \in \Sigma_f$. 

    We begin with the intermediate language $\tilde{L}_{x_1 \dots x_m} := L \cap \Sigma_\infty^* x_1 \Sigma_\infty^* x_2 \dots x_m \Sigma_\infty^*$. The intersection does not increase the tree width. Indeed, one has $\mathrm{TW}(L_1 \cap L_2) \leq \mathrm{TW}(L_1) \cdot \mathrm{TW}(L_2)$, and the automaton recognizing $\Sigma_\infty^* x_1 \Sigma_\infty^* x_2 \dots x_m \Sigma_\infty^*$ has tree width 1 whenever $x_i \in \Sigma_f$ for all $i$.

    Therefore, after trimming useless states, which can be done in time $O(dD^2)$ \cite{Sakarovitch2009}, we obtain an NFA for $\tilde{L}_{x_1\dots x_m}$ of size at most $(m+1)D$ and tree width at most $k$.\cite{Sakarovitch2009}.

    Next, we analyze the effect of substituting the symbols of $\Sigma_f$ by the single letter $f$ in $L_{x_1 \dots x_m} = T_{\Sigma_f \to f}(\tilde{L}_{x_1 \dots x_m})$. Replacing symbols in $\Sigma_\infty$ could introduce nondeterministic transitions on cycles (illustrated e.g. in Eq. \eqref{eq:nondeterm_transition-cycle}) and thus infinite tree width \cite[Cor. 3.1]{Palioudakis_2012_tree-width}. This cannot occur here: by definition of $\Sigma_f$, symbols in $\Sigma_f$ appear only a bounded number of times in any word of $L$, so no cycle of the automaton can carry a transition labeled by $\Sigma_f$. Consequently, substituting them by the single symbol $f$ cannot generate infinitely many partial computations, and the tree width remains finite. 

    More precisely, fix any string $w \in (\Sigma_\infty \cup \{f\})^*$, which we write without loss of generality as
    \begin{equation*}
        w = s_0 f s_1 f s_2 \dots f s_p
    \end{equation*}
    for some $p \in \mathbb{N}$, $y_i \in \Sigma_f$ and $s_i \in \Sigma_\infty^*$. The number of partial computations of $w$ in the automaton for $L_{x_1 \dots x_m}$ coincides with the number of partial computations of the string
    \begin{equation*}
        \tilde{w} :=
        \begin{cases}
            s_0 x_1 s_1 \dots x_p s_p & \text{if } p \leq m , \\ 
            s_0 x_1 s_1 \dots x_m s_m y_{m+1} s_{m+1} \dots y_p s_p & \text{if } p > m ,
        \end{cases}
    \end{equation*}
    in the trimmed automaton for $\tilde{L}_{x_1 \dots x_m}$, for any arbitrary choice of $y_{m+1}, \dots, y_p \in \Sigma_f$. Indeed, once a symbol different from the prescribed $x_i$ is encountered, the trimmed automaton has no outgoing transition and the computation halts. It follows that
    \begin{equation*}
        \mathrm{TW}(L_{x_1 \dots x_m}) \leq k .
    \end{equation*}

    By Lemma 8.1 of \cite{baburin_2023_on-the-fly-determinization-NFA}, determinizing the NFA accepting $L_{x_1 \dots x_m}$, of size $(m+1)D$ and tree width $\leq k$ yields a DFA with $O([(m+1)D)]^{k})$ states, constructible in time $O(d[(m+1)D)]^{1+k})$.

    In the next block of Algorithm \ref{alg:can_decomp}, each language $L_i^m$ is obtained as the intersection of $|\Sigma_f|^m$ automata, each of size at most $(m+1)D$ (for $L_{x_1 \dots x_m}$) or $O([(m+1)D]^{k})$ (for $L_{x_1 \dots x_m}^c$). The resulting automaton for $L_i^m$ therefore has size $O([(m+1)D]^{k})$ and can be computed in time $O(\prod_{i=1}^{d^m} t_i)$ where $t_i$ is the number of transitions of each automaton in the intersection. Since $t_i \leq O(d[(m+1)D]^{2k})$, the total computation time is $O((d[(m+1)D]^{2k})^{d^{m}})$. This must be repeated for all $S_i \in \mathcal{S}$, with $|\mathcal{S}| = 2^{|\Sigma_f|^m}$, so the total runtime of the block is $O((2d[(m+1)D]^{2k})^{d^m})$. Treating $m,d$ as constants, this scales as $O(D^{2kd^m})$. 

    Finally, performing the above steps for all $m \leq M$ yields an overall runtime of $O(M(2d[(M+1)D]^{2k})^{d^{M}})$, which scales with the bond dimension as $O(D^{2kd^M})$. In particular, for DFAs one has $k = 1$, so this becomes $O(D {2d^M})$, completing the proof.
    }
    
    \begin{algorithm}[h!] 
      \SetAlgoLined
      \caption{Canonical decomposition of a RL}
      \label{alg:can_decomp}
      \BlankLine
      \small
      \KwData{RL $L$ and any NFA representation of size $D$, $|\Sigma| = d$, and with at most $M$ appearances of symbols of $\Sigma_f$ in all words.}
      \KwResult{Compute the canonical decomposition $\{(L_j^m, X_j^m)\}_{j,m}$ of $L$ in time $O(M(d2^{1+2(M+1)D})^{d^M})$.}
      \BlankLine
      \Begin
      {
        \ForAll{$m \leq M$}{
            \ForEach{$x_1 x_2 \dots x_m \in (\Sigma_f)^m$}{
                $\tilde{L}_{x_1 \dots x_m} \gets L \cap \Sigma_\infty^* x_1 \Sigma_\infty^* x_2 \dots \Sigma_\infty^* x_m \Sigma_\infty^*$ \;

                $L_{x_1 \dots x_m} \gets T_{\Sigma_f \to f} (\tilde{L}_{x_1 \dots x_m})$ \;
                \lIf{$L_{x_1 \dots x_m} = \emptyset$}{\normalfont Discard it}
                \tcp{Runtime: \hspace{-2mm}$O(m d^2 D^2)$ (Prop. 1.4 and 1.11 in \cite{Sakarovitch2009})}

                $L_{x_1 \dots x_m}^c \gets$ Complement of $L_{x_1 \dots x_m}$ (first, determinize the NFA using Rabin-Scott powerset construction \cite{rabin_1959_powerset}, then calculate the complement of the DFA by changing $F \gets Q \setminus F$ \cite{Sakarovitch2009}
                \tcp*[l]{Runtime: \hspace{-2mm}$O(d (m+1) D2^{(m+1)D})$}
            }
            \tcp{Runtime of the block: \hspace{-2mm}$O(d^m 2^{mD})$}

            $\mathcal{S} \gets$ {\normalfont Power set of words $x_1 \dots x_m \in (\Sigma_f)^m$ such that $L_{x_1 \dots x_m} \neq \emptyset$ ($|\mathcal{S}| \leq 2^{|\Sigma_f|^m}$)}. 

            \ForEach{$S_i \in \mathcal{S}$}{\normalfont
                $L_i^m \gets$ Eq. \eqref{eq:sets_Li_intersection} (involves computing the intersection of $|\Sigma_f|^m$ NFA's of size either $\leq (m+1)D$ or $\leq 2^{(m+1)D}$, and checking emptiness) \;
                \tcp{Runtime: \hspace{-2mm}$O((d4^{(m+1)D})^{d^m})$}

                \If{$L_{x_1 \dots x_m} \neq \emptyset$}{$X_i^m \gets \cup_{x_1 \dots x_m \in \mathcal{S}_i} x_1 \dots x_m$ \;
                
                \tcp{Runtime: \hspace{-2mm}$O(d^m)$}
                }
            }
            \tcp{Runtime of the block: \hspace{-2mm} $O((d2^{1+2(m+1)D})^{d^m})$}  
        }
        \Return $\{(L_i^m, X_i^m)\}$ \;
        \tcp{Total runtime: \hspace{-2mm}$O(M(d2^{1+2(M+1)D})^{d^M})$}
    }
    \end{algorithm} 
\end{proof}



    {}

\section{Examples of the non-minimality of the canonical DFA} \label{app:E}

The minimal canonical DFA is not necessarily minimal with respect to all valid MPS representations of RLS, as has been studied in the RL literature \cite{Ravikumar_1989_succint-ambiguity}. However, if we consider a more general ansatz of the form of Eq. \eqref{eq:general_MPS-X} motivated by MPS theory, then we can obtain even smaller bond dimensions than the minimal UFA.

The first example illustrating this corresponds to $L = (0^*21^*3)^* \cup (0^*31^*2)^*$. While the Schmidt rank imposes a lower bound of 4 on any UFA representation of $L$, it can also be described by the following MPS-X with $D = 2$ and rank$(X) = 2$,
\begin{align*}{\footnotesize
        X = \begin{pmatrix}
            1 & \\
            & 1
        \end{pmatrix}, } \ 
    & {\footnotesize A^0 = \begin{pmatrix}
        1 & \\
        & 0
    \end{pmatrix}, \ 
    A^1 = \begin{pmatrix}
        0 & \\
        & 1
    \end{pmatrix}, \ } \\
    & {\footnotesize
    A^2 = \begin{pmatrix}
        0 & 1 \\
         & 0
    \end{pmatrix}, } \
    {\footnotesize A^3 = \begin{pmatrix}
        0 & \\
        1 & 0
    \end{pmatrix}. \ }
\end{align*}
Diagramatically, this can be seen as the automaton
\begin{equation*} \begin{tikzpicture}[scale=.2, baseline={([yshift=-3ex]current bounding box.center)}, thick]
        \tikzstyle{small state} = [state, minimum size=0pt, fill=purple,node distance=1.5cm,initial distance=1.5cm,initial text=$ $]
        \tikzset{->}

        \node[small state] (q1) {\small $q_1$};
        \node[small state, right of=q1] (q2) {\small $q_2$};
        
        \draw (q1) edge[loop above] node{\scriptsize $0$} (q1)
        (q2) edge[loop above] node{\scriptsize $1$} (q2)
        (q1) edge[bend left, above] node{\scriptsize $2$} (q2)
        (q2) edge[bend left, above] node{\scriptsize $3$} (q1);
    \end{tikzpicture}\ ,
    \end{equation*}
where $X$ tells us that any path through the automaton can either start at $q_1$ and accept at $q_1$, or start at $q_2$ and accept at $q_2$. 

As a second example, we consider $L = 0^\ast ( 1 0^\ast 2 0^\ast 3 \cup 2 0^\ast 3 0^\ast 1 \cup 3 0^\ast 1 0^\ast 2 ) 0^\ast$. The Schmidt rank gives a lower bound of 8 for the size of the minimal UFA, which is equal to the minimal DFA in this case,
\begin{equation*} \begin{tikzpicture}[scale=.175, baseline={([yshift=-3ex]current bounding box.center)}, thick]
    \tikzstyle{small state} = [state, minimum size=0pt, fill=purple,node distance=1.5cm,initial distance=1.5cm,initial text=$ $]
    \tikzset{->}

    \node[small state, initial] (q1) {\small $q_1$};
    \node[small state, right of=q1] (q4) {\small $q_4$};
    \node[small state, above of=q4] (q2) {\small $q_2$};
    \node[small state, below of=q4] (q6) {\small $q_6$};
    \node[small state, right of=q2] (q3) {\small $q_3$};
    \node[small state, right of=q4] (q5) {\small $q_5$};
    \node[small state, right of=q6] (q7) {\small $q_7$};
    \node[small state, accepting, right of=q5] (q8) {\small $q_8$};
    
    \draw   (q1) edge[loop above] node{\scriptsize $0$} (q1)
    (q2) edge[loop above] node{\scriptsize $0$} (q2)
    (q3) edge[loop above] node{\scriptsize $0$} (q3)
    (q4) edge[loop above] node{\scriptsize $0$} (q4)
    (q5) edge[loop above] node{\scriptsize $0$} (q5)
    (q6) edge[loop above] node{\scriptsize $0$} (q6)
    (q7) edge[loop above] node{\scriptsize $0$} (q7)
    (q8) edge[loop above] node{\scriptsize $0$} (q8)

    (q1) edge[above] node{\scriptsize $1$} (q2)
    (q2) edge[above] node{\scriptsize $2$} (q3)
    (q3) edge[above] node{\scriptsize $3$} (q8)
    
    (q1) edge[above] node{\scriptsize $2$} (q4)
    (q4) edge[above] node{\scriptsize $3$} (q5)
    (q5) edge[above] node{\scriptsize $1$} (q8)

    (q1) edge[above] node{\scriptsize $3$} (q6)
    (q6) edge[above] node{\scriptsize $1$} (q7)
    (q7) edge[above] node{\scriptsize $2$} (q8);
    
\end{tikzpicture}\ .
\end{equation*}

However, it can also be described by an MPS-X with $D = 6$ and rank$(X) = 3$ with
\begin{align*}
    & \footnotesize X = \dyad{4}{1} + \dyad{5}{2} + \dyad{6}{3}, \\
    & \footnotesize \quad \quad A^0 = \mathds{1}_6, \ 
    A^1 = \dyad{1}{2} + \dyad{4}{5}, \\
    & \footnotesize \quad \quad A^2 = \dyad{2}{3} + \dyad{5}{6}, \ 
    A^3 = \dyad{3}{4},
\end{align*}
which can be pictured as the automaton
\begin{equation*} \begin{tikzpicture}[scale=.2, baseline={([yshift=-3ex]current bounding box.center)}, thick]
    \tikzstyle{small state} = [state, minimum size=0pt, fill=purple,node distance=1.5cm,initial distance=1.5cm,initial text=$ $]
    \tikzset{->}

    \node[small state] (q1) {\small $q_1$};
    \node[small state, right of=q1] (q2) {\small $q_2$};
    \node[small state, right of=q2] (q3) {\small $q_3$};
    \node[small state, right of=q3] (q4) {\small $q_4$};
    \node[small state, right of=q4] (q5) {\small $q_5$};
    \node[small state, right of=q5] (q6) {\small $q_6$};
    
    \draw   (q1) edge[loop above] node{\scriptsize $0$} (q1)
    (q2) edge[loop above] node{\scriptsize $0$} (q2)
    (q3) edge[loop above] node{\scriptsize $0$} (q3)
    (q4) edge[loop above] node{\scriptsize $0$} (q4)
    (q5) edge[loop above] node{\scriptsize $0$} (q5)
    (q6) edge[loop above] node{\scriptsize $0$} (q6)

    (q1) edge[above] node{\scriptsize $1$} (q2)
    (q2) edge[above] node{\scriptsize $2$} (q3)
    (q3) edge[above] node{\scriptsize $3$} (q4)
    (q4) edge[above] node{\scriptsize $1$} (q5)
    (q5) edge[above] node{\scriptsize $2$} (q6);
\end{tikzpicture}\ ,
\end{equation*}
where $X$ tells us that any path through it can either start at $q_1$ and accept at $q_4$, start at $q_2$ and accept at $q_5$, or start at $q_3$ and accept at $q_6$.

\section{Proof of Theorem \ref{prop:LU} about the LU-equivalence of sparse RLS} \label{app:F}

A RL $L$ is sparse, with a number of words scaling as $O(N^k)$ with the word length $N$ for some natural number $k$, if and only if it can be described with a regex of the form 
\begin{equation} \label{eq:form_sparse}
    L = \bigcup_{i=1}^Q z_{i,0} (y_{i,1})^* z_{i,1} \dots (y_{i,t_i})^* z_{i,t_i} \ , 
\end{equation}
where $Q$ is a constant, $t_i \leq k+1$ for all $i$, and $y_i, z_i$ are strings in $\Sigma^*$ \cite{Szilard_1992}. Whether a RL is sparse or not can be decided in time $O(D + m)$, where $m$ is the number of transitions of an automaton accepting $L$ ($m \leq dD^2$) \cite{gawrychowski_2010_sparse_efficiently}. 

Note that, for a RL $L$ accepted by a UFA of size $D$, the number of words of fixed length $N$ in $L$ can be obtained as the norm of the associated MPS, as well as through the adjacency matrix of the UFA \cite{Chomsky_1958_structure_factor} in time $O(dD^2 + \log(N) D^3)$, or with a dynamic programming procedure \cite{kannan_1995_counting_words} in time $O(d N D^2)$.

Using the canonical representation in Corollary 6, we can characterize RLS that are equivalent under the action of $U^{\otimes N}$ for some unitary $U$.

\propLU*
\begin{proof}
    First, we assume that $| L_2^N \rangle = U^{\otimes} | L_1^N \rangle$ for all $N$. As the initial step, we show that $U$ acts as a permutation when restricted to $\Sigma_\infty$.

    Given $x \in \Sigma_\infty$, assume first for simplicity that it appears in $L_2$ in one of the strings $y_{i,j}$ of Eq. \eqref{eq:form_sparse} as $y_{i,j} = x$. Then, we know that for each $x \in \Sigma_\infty$ there exists a set of words $\{s_x^N\}_N \in L_2$ of the form $s_x^N = s_1 x^{N-a} s_2$, where $a := |s_1| + |s_2|$, and $s_1$, $s_2$ remain the same for all $N$. Then, 
    {\small\begin{align} 
        1 &= \langle s_i^N | U^{\otimes N} | L_1^N \rangle \nonumber \\
        &= \sum_{\substack{w_1w_2w_3 \in L_1 \cap \Sigma^N \\ |w_2| = N-a}} \mel{s_1}{U^{\otimes |w_1|}}{w_1} \mel{xx \dots x}{U^{\otimes (N-a)}}{w_2} \nonumber \\
        &\hspace{2.6cm}  \mel{s_2}{U^{\otimes |w_3|}}{w_3}. \label{eq:aux1_slocc}
    \end{align}}
    Assume that $|\mel{x}{U}{y}| \leq \lambda < 1$ for all $y \in \Sigma_\infty$. Since $L_1$ contains $O(\text{poly}(N))$ words of length $N$, Eq. \eqref{eq:aux1_slocc} would necessarily imply that $1 = O(\text{poly}(N) \lambda^N)$, which leads to a contradiction as $N \to \infty$. 
    
    Similarly, if $|y_{i,j}| > 1$ for all strings where $x \in \Sigma_\infty$ appears, the argument above still holds if we just choose $s_x^N = s_1 (y_{i,j})^{N-a} s_2$ in $L_2$, for any $y_{i,j}$ containing $x$. 
    
    Therefore, there is a permutation $\pi \in \mathcal{S}_{|\Sigma_\infty|}$ such that $U\ket{x} = \ket{\pi(x)}$ for $x \in \Sigma_\infty$, so $U = U_\pi \oplus U_f$, where $U_f \in \mathcal{U}_{|\Sigma_f|}$. The unitary $U$ thus acts on $L_1$ as
    {\begin{align*}
        U^{\otimes N} | L_1^N \rangle &= 
        \sum_m \sum_j \hat{S}^{(m)} \ \big( U_\pi^{\otimes N} | L_{1,j}^{m} \rangle \big) \big( U_f^{\otimes m} | X^{m}_{1,j} \rangle \big)
        \\
        &= \sum_m \sum_j \hat{S}^{(m)} \ | \pi( L_{1,j}^{m} )\rangle  \big( U_f^{\otimes m} | X^{m}_{1,j} \rangle \big)
        \\
        &= \sum_m \sum_j \hat{S}^{(m)} \ | L_{2,j}^{m,N} \rangle | X^{m}_{2,j} \rangle = | L_2^N \rangle.
    \end{align*}}   
    Due to the uniqueness of the canonical representation, we have that the above is only possible if, for each $m$, there exists some reordering of the elements in the canonical representation, such that for all $j$,
    \begin{equation} \label{eq:LU_aux2}
        \begin{cases}
            | L_{2,j}^{m,N} \rangle = | \pi\big(L_{1,j}^{m,N}\big) \rangle \\
            |X_{2,j}^{m}\rangle := U_f^{\otimes m} | X_{1,j}^{m} \rangle
        \end{cases}.
    \end{equation}
    The forward direction of the claim follows since Eq. \eqref{eq:LU_aux2} holds for all $N$.

    To prove the other direction, we just need to note that, defining $U := U_\pi \oplus U_f$, we would have $|L_2^N\rangle = U^{\otimes N} \ket{L_1^N}$ for all $N$ as desired. 
\end{proof}

\section{Examples of LU and SLOCC-equivalent RLS} \label{app:G}

First, we reproduce the example of LU-equivalent RLS in the main text and provide the explicit unitary relating them for completeness. Consider 
\begin{align*} \begin{cases}
    L_1 = 11 \cup 22 \cup 31 \cup 32 , \\
    L_2 = 11 \cup 12 \cup 32 \cup 33 , \end{cases}
\end{align*}
whose minimal DFAs have bond dimensions 5 and 4, respectively, as shown below,
\begin{align*}
    L_1 : 
    \begin{tikzpicture}[scale=.1, baseline={([yshift=0ex]current bounding box.center)}, thick]
    \tikzstyle{small state} = [state, minimum size=0pt, fill=purple,node distance=1.2cm,initial distance=1.5cm,initial text=$ $]
    \tikzset{->}
        \node[small state, initial] (q1) {\small $q_1$};
        \node[small state, right of=q1] (q3) {\small $q_3$};
        \node[small state, above of=q3] (q2) {\small $q_2$};
        \node[small state, below of=q3] (q4) {\small $q_4$};
        \node[small state, right of=q3, accepting] (q5) {\small $q_5$};
        \draw
        (q1) edge[above] node{\scriptsize $1$} (q2)
        (q1) edge[above] node{\scriptsize $2$} (q3)
        (q1) edge[above] node{\scriptsize $3$} (q4)
        (q2) edge[above] node{\scriptsize $1$} (q5)
        (q3) edge[above] node{\scriptsize $2$} (q5)
        (q4) edge[bend left, below] node{\scriptsize $1$} (q5)
        (q4) edge[bend right, above] node{\scriptsize $2$} (q5);
\end{tikzpicture}\ , \ 
    L_2 :  \begin{tikzpicture}[scale=.175, baseline={([yshift=0ex]current bounding box.center)}, thick]
    \tikzstyle{small state} = [state, minimum size=0pt, fill=purple,node distance=1.4cm,initial distance=1.5cm,initial text=$ $]
    \tikzset{->}
        \node[small state, initial] (q1) {\small $q_1$};
        \node[small state, above right of=q1] (q2) {\small $q_2$};
        \node[small state, below right of=q1] (q3) {\small $q_3$};
        \node[small state, below right of=q2, accepting] (q4) {\small $q_4$};
        \draw
        (q1) edge[above] node{\scriptsize $1$} (q2)
        (q1) edge[above] node{\scriptsize $3$} (q3)
        (q2) edge[bend left, below] node{\scriptsize $2$} (q4)
        (q2) edge[bend right, above] node{\scriptsize $1$} (q4)
        (q3) edge[bend left, below] node{\scriptsize $2$} (q4)
        (q3) edge[bend right, above] node{\scriptsize $3$} (q4);
\end{tikzpicture} \ .
\end{align*}
They are related as $\ket{L_2} = U^{\otimes 2} \ket{L_1}$ by the unitary $$U = \frac{1}{\sqrt{3}}{\footnotesize \begin{pmatrix} \frac{1}{2}  (1+\sqrt{3}) & \frac{1}{2}(1-\sqrt{3}) & 1 \\ 1 & 1 & -1 \\ \frac{1}{2}(1-\sqrt{3}) & \frac{1}{2}(1+\sqrt{3}) & 1 \end{pmatrix} }.$$

Another example of LU equivalent RLS involving Bell pairs consists of
\begin{equation*} \begin{cases}
    L_3 = 11 \cup 22, \\
    L_4 = 12 \cup 21,
\end{cases} \end{equation*}
whose minimal DFAs have size 2 as shown below, and therefore the bond dimension does not change in this case, 
\begin{align*}
    L_3 : 
    \begin{tikzpicture}[scale=.175, baseline={([yshift=0ex]current bounding box.center)}, thick]
    \tikzstyle{small state} = [state, minimum size=0pt, fill=purple,node distance=1.5cm,initial distance=1.5cm,initial text=$ $]
    \tikzset{->}
    \node[small state, initial] (q1) {\small $q_1$};
    \node[small state, above right of=q1] (q2) {\small $q_2$};
    \node[small state, below right of=q1] (q3) {\small $q_3$};
    \node[small state, below right of=q2, accepting] (q4) {\small $q_4$};
    \draw
    (q1) edge[above] node{\scriptsize $1$} (q2)
    (q1) edge[above] node{\scriptsize $2$} (q3)
    (q2) edge[above] node{\scriptsize $1$} (q4)
    (q3) edge[above] node{\scriptsize $2$} (q4)
    ;
\end{tikzpicture}\ , \ 
    L_4 :  \begin{tikzpicture}[scale=.175, baseline={([yshift=0ex]current bounding box.center)}, thick]
    \tikzstyle{small state} = [state, minimum size=0pt, fill=purple,node distance=1.5cm,initial distance=1.5cm,initial text=$ $]
    \tikzset{->}
    \node[small state, initial] (q1) {\small $q_1$};
    \node[small state, above right of=q1] (q2) {\small $q_2$};
    \node[small state, below right of=q1] (q3) {\small $q_3$};
    \node[small state, below right of=q2, accepting] (q4) {\small $q_4$};
    \draw
    (q1) edge[above] node{\scriptsize $1$} (q2)
    (q1) edge[above] node{\scriptsize $2$} (q3)
    (q2) edge[above] node{\scriptsize $2$} (q4)
    (q3) edge[above] node{\scriptsize $1$} (q4)
    ;
\end{tikzpicture} \ .
\end{align*}
They are related as $\ket{L_4} = U^{\otimes 2} \ket{L_3}$ by $U = \frac{1}{\sqrt{2}}{\footnotesize \begin{pmatrix} 1 & i \\ 1 & -i \end{pmatrix} }$.

Finally, we give an example involving the GHZ on three particles, where the minimal DFA of a language that is SLOCC-equivalent to it has a different number of internal states,
\begin{align*} \begin{cases}
    L_5 = 111 \cup 222 ,\\
    L_6 = 112 \cup 121 \cup 211 \cup 221 \cup 212 \cup 122 ,
\end{cases} \end{align*}
related as $\ket{L_6} = P^{\otimes 3} \ket{L_5}$ where
$$P = - {\footnotesize \begin{pmatrix} \frac{i}{3^{1/6}} & \left(-\frac{1}{3}\right)^{1/6} \\ \left(-\frac{1}{3}\right)^{1/6} & \frac{i}{3^{1/6}} \end{pmatrix} }.$$
Their minimal DFAs have bond dimensions 6 and 7, respectively, as shown below,
\begin{align*}
    L_5 : 
    &\begin{tikzpicture}[scale=.175, baseline={([yshift=0ex]current bounding box.center)}, thick]
    \tikzstyle{small state} = [state, minimum size=0pt, fill=purple,node distance=1.5cm,initial distance=1.5cm,initial text=$ $]
    \tikzset{->}
    \node[small state, initial] (q1) {\small $q_1$};
    \node[small state, above right of=q1] (q2) {\small $q_2$};
    \node[small state, right of=q2] (q4) {\small $q_4$};
    \node[small state, below right of=q4, accepting] (q6) {\small $q_6$};
    \node[small state, below right of=q1] (q3) {\small $q_3$};
    \node[small state, right of=q3] (q5) {\small $q_5$};
    \draw
    (q1) edge[above] node{\scriptsize $1$} (q2)
    (q2) edge[above] node{\scriptsize $1$} (q4)
    (q4) edge[above] node{\scriptsize $1$} (q6)
    (q1) edge[above] node{\scriptsize $2$} (q3)
    (q3) edge[above] node{\scriptsize $2$} (q5)
    (q5) edge[above] node{\scriptsize $2$} (q6);
\end{tikzpicture}\ , \\
    L_6 :  &\begin{tikzpicture}[scale=.175, baseline={([yshift=0ex]current bounding box.center)}, thick]
    \tikzstyle{small state} = [state, minimum size=0pt, fill=purple,node distance=1.5cm,initial distance=1.5cm,initial text=$ $]
    \tikzset{->}
    \node[small state, initial] (q1) {$q_1$};
    \node[small state, above right of=q1] (q2) {\small $q_2$};
    \node[small state, below right of=q1] (q3) {\small $q_3$};
    \node[small state, below right of=q2] (q5) {\small $q_5$};
    \node[small state, right of=q2] (q4) {\small $q_4$};
    \node[small state, right of=q3] (q6) {\small $q_6$};
    \node[small state, right of=q5, accepting] (q7) {\small $q_7$};
    \draw
    (q1) edge[above] node{\scriptsize $1$} (q2)
    (q1) edge[above] node{\scriptsize $2$} (q3)
    (q2) edge[above] node{\scriptsize $1$} (q4)
    (q2) edge[above] node{\scriptsize $2$} (q5)
    (q3) edge[above] node{\scriptsize $1$} (q5)
    (q3) edge[above] node{\scriptsize $2$} (q6)
    (q4) edge[above] node{\scriptsize $2$} (q7)
    (q5) edge[above] node{\scriptsize $1,2$} (q7)
    (q6) edge[above] node{\scriptsize $1$} (q7);
\end{tikzpicture} \ . 
\end{align*}

\section{Generalization of RLS to 2D regular pictures} \label{app:H}

For completeness, we provide here more details on the generalization of RLs to two dimensions following the notation in \cite{Giammarresi_1997}, and their explicit expression in terms of PEPS. 

A \textit{picture} can be defined as a 2D rectangular array of elements of an alphabet $\Sigma$. A \textit{2D language} is then defined as a subset of $\Sigma^{**}$, where $\Sigma^{**}$ denotes the set of all pictures over $\Sigma$. 1D RL's thus correspond to pictures with a single row or column.

\begin{definitionnon}[4.2 in \cite{Giammarresi_1997}] A (non-deterministic) \textit{2D online tessellation automata} (OTA) is defined by $\mathcal{A} = \langle Q, \Sigma, \delta, q_0, F \rangle$, where $Q$ is the set of internal states, $\Sigma$ is the input alphabet, $q_0$ is the initial state, $F$ is the set of accepting states, and $\delta: Q \times Q \times \Sigma \to \mathcal{P}(Q)$ is the transition function.
\end{definitionnon}

A \textit{run} of $\mathcal{A}$ on picture $p \in \Sigma^{**}$ consists of associating a state from $Q$ to each element $(i,j)$ of $p$. The state is determined by the transition function $\delta$ which receives as inputs the states already associated to positions $(i-1,j)$ and $(i,j-1)$, and the symbol $p(i,j)$. At time $t = 0$, the internal states of the first row and the first column are set to the initial state indicated by $q_0$. 

Then, we say that an OTA $\mathcal{A}$ \textit{recognizes} or \textit{accepts} a picture $p$ of size $(l_1(p), l_2(p))$ denoting its number of rows and columns, respectively, if there exists a run of $\mathcal{A}$ on $p$ such that the state associated to position $(l_1(p), l_2(p))$ is an accepting state belonging to $F$.

In fact, we can express the above definition of OTA as a PEPS, with physical dimension equal to $|\Sigma|$, bond dimension $D = |Q|$, and the following tensors:
\begin{align*}
    \begin{tikzpicture}[scale=.2, baseline={([yshift=-0.5ex]current bounding box.center)}, thick]
        \pepstensor{0,0}{}{purple}{$x$}{$z$}{$z$}{$y$}{$a$}
    \end{tikzpicture} &
    = \begin{cases}
        1 \text{ if } z \in \delta(x,y,a), \\
        0 \text{ otherwise,}
    \end{cases} \\
    \begin{tikzpicture}[scale=.2, baseline={([yshift=-0.5ex]current bounding box.center)}, thick]
        \boundarypeps{0,0}{amaranth}
    \end{tikzpicture} &= \ket{q_0} , \
    \begin{tikzpicture}[scale=.2, baseline={([yshift=-0.5ex]current bounding box.center)}, thick]
        \boundarypeps{0,0}{yellow}
    \end{tikzpicture} = \sum_{f \in F} \ket{f} ,
\end{align*}
so that the PEPS evaluated on picture $p$ is
\begin{equation*}
    \begin{tikzpicture}[scale=.2, baseline={([yshift=-0.5ex]current bounding box.center)}, thick]
        \fullpepstensor{0,0};
    \end{tikzpicture}
    \begin{cases}
        > 0 \text{ if } p \in L, \\
        = 0 \text{ if } p \notin L.
    \end{cases}
\end{equation*}

\marta{By construction, all these states satisfy an area law, with the entanglement upper bounded by $|\partial A|\log \vert Q \vert$, $|\partial A|$ being the number of virtual bonds cut by the boundary of the considered region $A$.}

As an example, we provide the explicit representation of a 2D language over $\Sigma = \{a\}$ consisting of all pictures with an odd number of columns. An OTA that recognizes it is given by (4.3 in \cite{Giammarresi_1997}):
\begin{itemize}
    \item $Q = \{0, 1, 2\}$,
    \item $q_0 = \{0\}$,
    \item $F = \{1\}$,
    \item $\delta(0,0,a) = \delta(0,2,a) = \delta(1,0,a) = \delta(1,2,a) = 1$; $\delta(0,1,a) = \delta(2,1,a) = 2$.
\end{itemize}
The equivalent PEPS representation has $d = 1$, $D = 3$, and tensors
\begin{align*}
    \begin{tikzpicture}[scale=.2, baseline={([yshift=-0.5ex]current bounding box.center)}, thick]
        \pepstensor{0,0}{}{purple}{$0$}{$1$}{$1$}{$0$}{$a$};
    \end{tikzpicture}
    &= \begin{tikzpicture}[scale=.2, baseline={([yshift=-0.5ex]current bounding box.center)}, thick]
        \pepstensor{0,0}{}{purple}{$0$}{$1$}{$1$}{$2$}{$a$};
    \end{tikzpicture}
    = \begin{tikzpicture}[scale=.2, baseline={([yshift=-0.5ex]current bounding box.center)}, thick]
        \pepstensor{0,0}{}{purple}{$1$}{$1$}{$1$}{$0$}{$a$};
    \end{tikzpicture}
    = \begin{tikzpicture}[scale=.2, baseline={([yshift=-0.5ex]current bounding box.center)}, thick]
        \pepstensor{0,0}{}{purple}{$1$}{$1$}{$1$}{$2$}{$a$};
    \end{tikzpicture} 
    = 1 \\
    \begin{tikzpicture}[scale=.2, baseline={([yshift=-0.5ex]current bounding box.center)}, thick]
        \pepstensor{0,0}{}{purple}{$0$}{$2$}{$2$}{$1$}{$a$};
    \end{tikzpicture}
    &= \begin{tikzpicture}[scale=.2, baseline={([yshift=-0.5ex]current bounding box.center)}, thick]
        \pepstensor{0,0}{}{purple}{$2$}{$2$}{$2$}{$1$}{$a$};
    \end{tikzpicture} 
    = 1 \\
    \begin{tikzpicture}[scale=.2, baseline={([yshift=-0.5ex]current bounding box.center)}, thick]
        \boundarypeps{0,0}{amaranth};
    \end{tikzpicture} &= \ket{0} , 
    \begin{tikzpicture}[scale=.2, baseline={([yshift=-0.5ex]current bounding box.center)}, thick]
        \boundarypeps{0,0}{yellow};
    \end{tikzpicture} = \ket{1} .
\end{align*}
